\documentclass[journal,12pt,onecolumn,draftclsnofoot,]{IEEEtran}

\usepackage[T1]{fontenc}
\usepackage{subcaption}
\usepackage{dsfont}
\usepackage{color}
\usepackage{varioref}
\usepackage{textcomp}
\usepackage{amsthm}
\usepackage{amsmath}
\usepackage{amssymb}
\usepackage{graphicx}
\usepackage{setspace}
\usepackage{esint}
\usepackage{algorithm}
\usepackage{algpseudocode}
\usepackage{pifont}
\usepackage{wrapfig}
\usepackage{multirow}
\usepackage{tabu}
\usepackage{amsmath}

\usepackage{enumitem}
\usepackage{cite}
\usepackage{cleveref}
\usepackage{lmodern}
\usepackage{amssymb}
\usepackage{graphicx,scalerel}
\newcommand\wh[1]{\hstretch{2}{\hat{\hstretch{.5}{#1}}}}

\makeatletter

\newcommand{\Rmnum}[1]{\expandafter\@slowromancap\romannumeral #1@}
\makeatother

\newtheorem{theorem}{Theorem}
\newtheorem{lemma}{Lemma}
\newtheorem{remark}{Remark}

\newtheorem{assumption}{Assumption}
\newtheorem{corollary}{Corollary}
\theoremstyle{definition}

\providecommand{\propositionname}{Proposition}

\ifCLASSINFOpdf

\else

\fi

\usepackage[T1]{fontenc}
\usepackage{etoolbox}

\makeatletter
\patchcmd{\maketitle}{\@fnsymbol}{\@alph}{}{}  % Footnote numbers from symbols to small letters
\makeatother

\title{Convergence of Federated Learning over \\a Noisy Downlink}
\author{\IEEEauthorblockN{Mohammad Mohammadi Amiri, Deniz G\"und\"uz, Sanjeev R. Kulkarni,\\ H. Vincent Poor\thanks{M. Mohammadi Amiri, S. R. Kulkarni, and H. V. Poor are with the Department of Electrical Engineering, Princeton University, Princeton, NJ 08544, USA (e-mail: \{mamiri, kulkarni, poor\}@princeton.edu).} \thanks{D. G\"und\"uz is with the Department of Electrical and Electronic Engineering, Imperial College London, London SW7 2AZ, U.K. (e-mail: d.gunduz@imperial.ac.uk).}
}
}

\begin{document}

\maketitle

%\hspace{-0.5in}
%\thispagestyle{empty}
\begin{abstract}
We study federated learning (FL), where power-limited wireless devices utilize their local datasets to collaboratively train a global model with the help of a remote parameter server (PS). 
The PS has access to the global model and shares it with the devices for local training using their datasets, and the devices return the result of their local updates to the PS to update the global model.\makeatletter{\renewcommand*{\@makefnmark}{}\footnotetext{This work was supported in part by the U.S. National Science Foundation under Grant CCF-0939370, and by the European Research Council (ERC) Starting Grant BEACON (grant agreement no. 677854).}\makeatother} 
The algorithm continues until the convergence of the global model. 
This framework requires downlink transmission from the PS to the devices and uplink transmission from the devices to the PS.
%, where we assume wireless shared mediums for both uplink and downlink. 
The goal of this study is to investigate the impact of the bandwidth-limited shared wireless medium in both the downlink and uplink on the performance of FL with a focus on the downlink.
To this end, the downlink and uplink channels are modeled as fading broadcast and multiple access channels, respectively, both with limited bandwidth. 
For downlink transmission, we first introduce a digital approach, where a quantization technique is employed at the PS followed by a capacity achieving channel code to transmit the global model update over the wireless broadcast channel at a common rate such that all the devices can decode it.
Next, we propose analog downlink transmission, where the global model is broadcast by the PS in an uncoded manner. 
We consider analog transmission over the uplink in both cases, since its superiority over digital transmission for uplink has been well studied in the literature.
We further analyze the convergence behavior of the proposed analog transmission approach over the downlink assuming that the uplink transmission is error-free. 
Numerical experiments show that the analog downlink approach provides significant improvement over the digital one, despite a significantly lower transmit power at the PS, with a more notable improvement when the data distribution across the devices is not independent and identically distributed. 
The experimental results corroborate the convergence results, and show that a smaller number of local iterations should be used when the data distribution is more biased, and also when the devices have a better estimate of the global model in the analog downlink approach.
\end{abstract}

%\begin{IEEEkeywords}
%Gaussian broadcast channels, decentralized caching, superposition coding.
%\end{IEEEkeywords}

%\newpage
\section{Introduction}\label{SecIntro}

%%%%% Longer version

Wireless devices, such as mobile phones, wearables, and Internet-of-things (IoT) devices, continuously generate massive amounts of data. 
This massive data can be processed to infer the state of a system, or to anticipate its future states with applications in autonomous driving, unmanned aerial vehicles (UAVs), or extended reality (XR) technologies. Due to the growing storage and computational capabilities of wireless edge devices, it is increasingly attractive to store and process the data locally by shifting network computations to the edge. 
Also, in contrast to traditional machine learning (ML) solutions, it is not desirable to offload such massive amounts of data available at the wireless edge devices to a cloud server for centralized processing due to latency, bandwidth, and power constraints in wireless networks, as well as privacy concerns of users.
\textit{Federated learning} (FL) has emerged as an alternative method enabling ML at the wireless network edge by utilizing wireless edge computational capabilities to process data locally.

%Today devices at the wireless network edge generate a huge amount of data that can be exploited to make sense of the state of a system or predict its future states. Internet-of-things (IoT) devices, unmanned aerial vehicles (UAVs), or extended reality (XR) technologies are prime examples, where data from multiple sensors must be continuously collected and processed. Machine learning (ML) algorithms are being developed to exploit these massive datasets. Most current ML solutions focus on centralized algorithms, where a cloud server collects all the data to train a powerful model. However, offloading such massive amounts of data from the edge devices to the could server is often not feasible due to latency, bandwidth, or power constraints, or not allowed due to privacy concerns. A more desirable and practically viable approach is \textit{federated learning} (FL), which enables ML at the wireless edge while the data never leaves the edge devices.

In FL the goal is to fit a global model to data generated and stored locally at the wireless devices by exploiting edge processing capabilities collaboratively with the help of a remote parameter server (PS) \cite{DCKonecnyFederated}.
The PS keeps track of the global model, which is updated using the local model updates received from the participating devices, and shares it with the devices for training using their local data.
When FL is employed at the wireless edge, the PS can be a wireless access point or a base station, and the communication between the PS and the devices takes place over the shared wireless medium with limited energy and bandwidth. 
There have been several studies to develop distributed ML techniques with communication constraints \cite{DCKonecnyFederated,McMahan2017CommunicationEfficientLO,GoogleMcMahanFed,KonecnyRandDistMean,SmithFedMultiTask,KonecnyFLBeyondData,NishioFLClientSelecHetRes,FLWithnonIIDZhao,XLiFedAveFLnonIID,COLAFLHe,AgonsticFLMohri}. 
However, these studies focus on limiting the uplink communication from the devices to the PS by assuming rate-limited error-free links, and do not take into consideration the physical layer characteristics of the wireless medium.

Recently there have been efforts to develop a federated edge learning (FEEL) framework considering the physical layer aspects of the underlying wireless medium. 
FL over power- and bandwidth-limited multiple access channel (MAC) for the uplink is studied in \cite{MohammadDenizDSGDCS}, and novel digital and analog transmission techniques at the wireless devices are proposed. 
While the former employs gradient sparsification followed by quantization and channel coding for digital transmission, the latter utilizes the superposition property of the underlying wireless MAC, and introduces a novel bandwidth-efficient transmission technique employing sparsification and linear projection. 
FL over a broadband wireless fading MAC is studied in \cite{KaibinParallelWork}, where the devices have channel state information (CSI) to perform channel inversion, while \cite{CohenAnalogGDDL} proposes analog transmission over the wireless fading MAC without any power control. 
The extension of the approach introduced in \cite{MohammadDenizDSGDCS} to the wireless fading MAC studied in \cite{MohamamdDenizFLOverAirSPAWC19,FLTWCMohammadDenizFading}, which combines the linear projection idea of \cite{MohammadDenizDSGDCS} with power control. 
%Analog transmission over wireless fading MAC has also found its way in distributed ML algorithms \cite{CohenAnalogGDDL}.  
Furthermore, FL over wireless networks with a multi-antenna PS is studied in \cite{YangFedLearOverAirComp,TungVuFLMassiveMIMO,MohammadTolgaDenizFLGlobalSIP,YoSebMohammadMIMOFL}, where beamforming techniques are used for efficient gradient aggregation at the PS. 
In \cite{RaviCommEffFLGaussian} digital transmission over a Gaussian MAC from the devices to the PS is considered with quantization based on the channel qualities, and \cite{DenizOneBitAgg} studies digital transmission using the over-the-air aggregation property of the wireless MAC.  
Various device scheduling policies are studied for FEEL aiming to select a subset of the devices sharing the limited wireless resources efficiently, including frequency of participation in the training \cite{YangArafaVinceAgeBasedFL}, minimizing the training delay \cite{ShiZhouFastConverFL}, link qualities of the devices \cite{HowardVinceSchedulingsFL}, energy consumption \cite{YuxuanFLRedundantData}, and importance of the model update along with the channel quality \cite{MohammadDenizSanjVinceISIT20}. 
Resource allocation for FEEL is formulated as an optimization problem to speed up training \cite{AccelDNNFLEdge}, to minimize the empirical loss function \cite{ChenVinceJointLearCommFL}, and to minimize the total energy consumption \cite{DinhFLWireNetConver}.
Also, convergence of FEEL with limited bandwidth from the devices to the PS is analyzed in \cite{MohammadFLConvergence}.

All the aforementioned works assume an error-free PS-to-devices shared link, and availability of an accurate global model at the devices for local training. 
In this paper, we consider a bandwidth-limited wireless fading broadcast channel from the PS to the devices with limited transmit power at the PS. 
We introduce \textit{digital} and \textit{analog} transmission approaches over the downlink. 
In the digital downlink, the PS employs quantization followed by channel coding to broadcast the quantized global model update over the wireless fading broadcast channel, at a rate targeting the device with the worst channel, so that all the devices can successfully receive the global model. 
On the other hand, with the analog downlink approach, the PS broadcasts the global model vector in an analog/uncoded manner over the wireless fading broadcast channel, and the devices receive different noisy versions of it. 
We model the uplink from the devices to the PS, over which the devices send their model updates, as a bandwidth-limited fading MAC. 
We follow the existing works highlighting the efficiency of the analog transmission over the uplink fading MAC for FEEL \cite{MohammadDenizDSGDCS,KaibinParallelWork,FLTWCMohammadDenizFading}, and consider analog communications.
The convergence analysis of the proposed digital downlink approach is provided in \cite{MDSV_NIPS2020}.
%utilizing the additive nature of the MAC  
Here, we provide the convergence analysis of the analog downlink approach, where for ease of analysis we assume error-free uplink transmission and focus on the impact of a noisy downlink transmission on the convergence behavior. 
Our theoretical analysis is complemented with numerical experiments on the MNIST dataset, which clearly illustrate the significant advantages of the analog downlink approach compared to its digital counterpart. 
We observe that the improvement is more significant when the data is not independent and identically distributed (iid) across the devices. 
The performance of both approaches improve with the number of devices thanks to the additional power introduced by each device.
Our numerical results corroborate the analytical convergence analysis, showing that reducing the number of local iterations provides the best performance when introducing bias in the data distribution across the devices. 
Also, both analytical and experimental results show that, for non-iid data distribution, the number of local iterations at the devices should reduce when the transmit power at the PS increases.

Imperfect downlink transmission in FL is also treated in \cite{ExpandingRedClientResFL} and \cite{JinHyunAhnFLNoisyDownlink}. In \cite{ExpandingRedClientResFL}, the shared link from the PS to the devices is assumed to be rate-limited without taking into account the physical layer characteristics of the wireless medium; 
the PS sends a compressed version of the current global model to the devices through quantization. 
The efficiency of quantizing the global model diminishes significantly since the peak-to-average ratio of the parameters is high. 
Therefore, \cite{ExpandingRedClientResFL} proposes employing a linear projection at the PS to first spread the information of the global model vector more evenly across its dimensions, and the devices perform the inverse of the linear projection to estimate the global model vector. 
Instead, in our proposed digital downlink approach, the PS broadcasts the quantized global model update, with respect to the global model estimate at the devices, and the devices recover an estimate of the current global model using their knowledge of the last global model. 
We highlight that the global model update has significantly less variability/variance than the global model itself. 
Hence, compared to the proposed digital downlink approach, the approach in \cite{ExpandingRedClientResFL} requires significantly higher computation overhead at the PS and the devices due to the linear projection and its inverse, respectively, and this overhead grows with the number of model parameters. 
Moreover, the results in both \cite{ExpandingRedClientResFL} and \cite{JinHyunAhnFLNoisyDownlink} are limited to simulations, where \cite{JinHyunAhnFLNoisyDownlink} illustrates the advantages of analog transmission in the downlink but does not provide a convergence result. 
In this paper, we provide an in-depth analysis of the impact of a noisy downlink on the performance of FEEL through extensive experimental results together with theoretical convergence analysis.

The rest of this paper is organized as follows. In Section \ref{SecProbFormul}, we present the system model. The digital and analog downlink approaches are introduced in Section \ref{SecDigDown} and Section \ref{SecAnalog}, respectively. In Section \ref{SecConvergence}, we provide the convergence results of the analog downlink approach. Numerical results are presented in Section \ref{SecExperiments}. Finally, we conclude the paper in Section \ref{SecConc}, and provide a detailed proof of the main theorem in the Appendices.

%%%%%% Longer version
%In this paper, we consider FL with digital transmission from the edge devices to the PS over a block fading wireless network with limited resources, where we design novel device scheduling policies deciding the devices participating in each round. We design resource allocation across the participating devices to perform orthogonal (interference-free) transmission. Due to resource limitations, we develop device scheduling policies taking into account the channel conditions and the significance of the local model updates at the devices to make sure that the resources are allocated across the devices with important messages and proper link capacity to convey their messages. Numerical results illustrate the advantages of considering both the channel conditions and the local model updates at the devices for device scheduling over scheduling based on either of the two metrics individually.        

%\subsection{Notations}\label{SecNot}
\textit{Notation}: We denote the set of real, natural and complex numbers by $\mathbb{R}$, $\mathbb{N}$ and $\mathbb{C}$, respectively. 
For $i \in \mathbb{N}$, we let $[i] \triangleq \{ 1, \dots, i \}$. 
We denote a circularly symmetric complex Gaussian distribution with real and imaginary components with variance $\sigma/2$ by $\mathcal{C N} \left( 0,\sigma \right)$. 
For vectors $\boldsymbol{x}$ and $\boldsymbol{y}$ with the same dimension, $\boldsymbol{x} \circ \boldsymbol{y}$ returns their Hadamard/entry-wise product. 
Also, ${\rm{Re}} \{ \boldsymbol{x} \}$ and ${\rm{Im}} \{ \boldsymbol{x} \}$ return entry-wise real and imaginary components of $\boldsymbol{x}$, respectively, and $(\boldsymbol{x})^{-1}$ represents entry-wise inverse of vector $\boldsymbol{x}$. 
The notation $\left| \cdot \right|$ represents the cardinality of a set, the $l_2$-norm of vector $\boldsymbol{x}$ is denoted by $\left\| \boldsymbol{x} \right\|_2$, and $\langle \boldsymbol{x}, \boldsymbol{y} \rangle$ denotes the inner product of vectors $\boldsymbol{x}$ and $\boldsymbol{y}$. The imaginary unit is represented by $j$. 
%For a set $\cal S$ of real numbers, ${\max}_{[K]} \mathcal S$ returns the $K$-element subset of $\mathcal S$ with the highest values.   

%\begin{figure}[!t]
%\centering
%\includegraphics[scale=0.7]{System_Model_Cropped.pdf}
%\caption{Illustration of distributed machine learning at the wireless edge. Workers with limited local datasets collaborate through a PS to implement DSGD at the network edge, where the local gradient estimates of the workers are transmitted to the PS over a shared wireless MAC.} 
%\label{System_Model}
%\end{figure}

%\begin{figure}[!t]
%\centering
%\includegraphics[scale=0.3,trim={175pt 125pt 260pt 87pt},clip]{System_model.eps}
%\caption{Illustration of distributed machine learning at the wireless edge. Workers with limited local datasets collaborate through a PS to implement DSGD at the network edge, where the local gradient estimates of the workers are transmitted to the PS over a shared wireless MAC.} 
%\label{System_Model}
%\end{figure}

\section{System Model}\label{SecProbFormul}

We consider FEEL where $M$ wireless devices collaboratively train a model parameter vector $\boldsymbol{\theta} \in \mathbb{R}^d$ with the help of a remote parameter server (PS). Device $m$ has access to $B_m$ local data samples, the set of which is denoted by $\mathcal{B}_m$, i.e., $B_m = \left| \mathcal{B}_m \right|$, $m \in [M]$, and we define $B \triangleq \sum\nolimits_{m=1}^{M} B_m$. The goal is to minimize loss function 
\begin{align}\label{GenEmpLossFunc}
F \left( \boldsymbol{\theta} \right) =  \sum\nolimits_{m=1}^{M} \frac{B_m}{B} F_m \left(\boldsymbol{\theta} \right),   
\end{align}
where $F_m \left(\boldsymbol{\theta} \right)$ denotes the loss function at device $m$, 
\begin{align}
F_m \left( \boldsymbol{\theta} \right) = \frac{1}{B_m} \sum\nolimits_{\boldsymbol{u} \in \mathcal{B}_m} f \left(\boldsymbol{\theta}, \boldsymbol{u} \right), \quad m \in [M],   
\end{align}
where $f(\cdot, \cdot)$ is an empirical loss function defined by the learning task.
Device $m$ performs multiple iterations of stochastic gradient descent (SGD) algorithm based on its local dataset and the global model parameter vector shared by the PS to minimize $F_m \left( \boldsymbol{\theta} \right)$, $m \in [M]$.

FEEL involves iterative communications between the wireless devices and the PS until the model parameter vector converges to its optimum, minimizing loss function $F(\boldsymbol{\theta})$. It consists of \textit{downlink} and \textit{uplink} wireless transmissions, where in the downlink the PS shares the global model parameter vector with the devices for local training, and in the uplink the devices transmit their local model updates to the PS, which updates the global model parameter vector accordingly. 
%In the following, we model the downlink and uplink wireless channels for an FL framework. 

During the $t$-th global iteration, the PS broadcasts the global model parameter vector, denoted by $\boldsymbol{\theta} (t)$, to the devices over the downlink channel. 
We model the downlink wireless channel as a fading broadcast channel, where OFDM with $n^{\rm{dl}}$ subchannels is employed for transmission. 
We denote the length-$n^{\rm{dl}}$ channel input by the PS at the global iteration $t$ by $\boldsymbol{x}^{\rm{dl}} (t) \in \mathbb{C}^{n^{\rm{dl}}}$, and consider a transmit power constraint $P^{\rm{dl}}$ at the PS at any global iteration. 
The received signal at device $m$ is given by
\begin{align}
\boldsymbol{y}_{m}^{\rm{dl}} (t) =  \boldsymbol{h}^{\rm{dl}}_m (t) \circ \boldsymbol{x}^{\rm{dl}} (t) +  \boldsymbol{z}_m^{\rm{dl}} (t), \quad \mbox{for $m \in [M]$},   
\end{align}
where $\boldsymbol{h}^{\rm{dl}}_m (t) \in \mathbb{C}^{n^{\rm{dl}}}$ is the downlink channel gain vector from the PS to device $m$ with each entry iid according to $\mathcal{CN} (0, \sigma^{\rm{dl}})$, and $\boldsymbol{z}^{\rm{dl}}_m (t) \in \mathbb{C}^{n^{\rm{dl}}}$ is the downlink additive noise vector at device $m$ with each entry iid according to $\mathcal{CN} (0, 1)$. 
We assume that device $m$ has channel state information (CSI) about the downlink channel, and denote the noisy estimate of the global model parameter vector $\boldsymbol{\theta} (t)$ at device $m$ by $\wh{\boldsymbol{\theta}}_m (t)$, $m \in [M]$.

Having estimated $\wh{\boldsymbol{\theta}}_m (t)$, device $m$, $m \in [M]$, updates the model by running SGD $\tau$ steps locally, for some $\tau \in \mathbb{N}$. 
The $i$-th SGD step at device $m$ during global iteration $t$ is given by
\begin{align}
\boldsymbol{\theta}_m^{i+1} (t) = \boldsymbol{\theta}_m^i (t) - \eta^i_m (t) \nabla F_m \left( \boldsymbol{\theta}_m^i (t), \xi_m^i (t) \right),  \quad \mbox{$i \in [\tau]$},   
\end{align}
where $\boldsymbol{\theta}_m^1 (t) = \wh{\boldsymbol{\theta}}_m (t)$, $\eta^i_m (t)$ represents the learning rate, and $\nabla F_m \left( \boldsymbol{\theta}_m^i (t), \xi_m^i (t)  \right)$ denotes the stochastic gradient estimate with respect to $\boldsymbol{\theta}_m^i (t)$ and the local mini-batch sample $\xi_m^i (t)$, chosen uniformly at random from the local dataset $\mathcal{B}_m$, for $m \in [M]$. 
We highlight that 
\begin{align}\label{AverageStochGradientEst}
\mathbb{E}_{\xi} \left[ \nabla F_m \left( \boldsymbol{\theta}_m^i (t), \xi_m^i (t) \right) \right] = \nabla F_m \left( \boldsymbol{\theta}_m^i (t)  \right), \quad \forall i \in [\tau], \forall m \in [M], \forall t,   
\end{align}
where $\mathbb{E}_{\xi}$ denotes expectation with respect to the randomness of the stochastic gradient function. After performing the local SGD algorithm, device $m$ aims to transmit the local model update $\Delta \boldsymbol{\theta}_m (t) = \boldsymbol{\theta}_m^{\tau+1} (t) - {\boldsymbol{\theta}}_m^1 (t)$ to the PS over the uplink channel, $m \in [M]$.

We model the uplink channel as a fading MAC, where, similarly to the downlink, OFDM is employed for transmission. We assume $n^{\rm{up}}$ subchannels are available to each device in the uplink with transmit power constraint $P^{\rm{up}}$ during each global iteration.
The length-$n^{\rm{up}}$ channel input by device $m$ at the global iteration $t$ is denoted by $\boldsymbol{x}^{\rm{up}}_m (t) \in \mathbb{C}^{n^{\rm{up}}}$, for $m \in [M]$. 
The channel output received at the PS during the global iteration $t$ is given by
\begin{align}
\boldsymbol{y}^{\rm{up}} (t) = \sum\nolimits_{m =1}^{M} \boldsymbol{h}^{\rm{up}}_{m} (t) \circ \boldsymbol{x}^{\rm{up}}_{m} (t) + \boldsymbol{z}^{\rm{up}} (t),
\end{align}
where $\boldsymbol{h}^{\rm{up}}_m (t) \in \mathbb{C}^{n^{\rm{up}}}$ is the uplink channel gain vector from device $m$ to the PS with each entry iid according to $\mathcal{CN} (0, \sigma^{\rm{up}})$, and $\boldsymbol{z}^{\rm{up}}_m (t) \in \mathbb{C}^{n^{\rm{up}}}$ is the uplink additive noise vector at the PS with each entry iid according to $\mathcal{CN} (0, 1)$. 
We assume that the PS knows all the channel gains, while each device knows the states of its own subchannels. 
The PS's goal is to recover the average of the local model updates, $\frac{1}{M} \sum\nolimits_{m=1}^{M} \Delta \boldsymbol{\theta}_m (t)$, whose estimate at the PS is denoted by $\Delta \wh{\boldsymbol{\theta}} (t)$, which is then used to obtain the updated global model parameter vector, $\boldsymbol{\theta} (t+1)$.

In this paper, we study the impact of noisy downlink transmission on the performance of FEEL. For this purpose, we consider digital and analog transmission approaches over the downlink channel. 
When performing digital transmission, we assume that the PS has CSI about the downlink wireless channels, while for the analog transmission, no CSI about the downlink channels at the PS is needed. 
On the other hand, following the results in \cite{MohammadDenizDSGDCS,FLTWCMohammadDenizFading,KaibinParallelWork}, which have shown the superiority of analog transmission for the uplink transmission over a wireless MAC, here we only consider analog transmission over the uplink.

\section{Digital Downlink Approach}\label{SecDigDown}

%at iteration $t$, 

In this section, we present a digital approach for the downlink transmission of the global model update to the devices. 
%We then study various scheduling policies limiting the number of devices sharing the resources for transmission at each iteration.
%\subsection{Digital SGD}\label{SubSecDSGD}
%We utilize the technique introduced in \cite{DCSattlerSparseBinary}, and extended in \cite{MohammadDenizDSGDCS} for FL over a bandwidth-limited wireless medium. 

\subsection{Downlink Channel Capacity}\label{SubSecDownCap} 
At the global iteration $t$, the PS aims to transmit vector $\boldsymbol{x}^{\rm{dl}} (t)$, containing information about the global model vector $\boldsymbol{\theta} (t)$, to all the devices using digital transmission with transmit power ${P}^{\rm{dl}}$ over the bandwidth-limited wireless channel. 
The PS broadcasts $\boldsymbol{x}^{\rm{dl}} (t)$ at a ``common rate'' such that all the devices can decode it. 
The downlink is a parallel fading broadcast channel with $n^{\rm{dl}}$ subchannels, where CSI is known at both the transmitter and the receivers. 
In the following, we provide an upper bound on the maximum common rate of broadcasting over this $n^{\rm{dl}}$ parallel fading channels. 
Given an average transmission power ${P}^{\rm{dl}}$ at global iteration $t$, the maximum common rate of downlink transmission over $n^{\rm{dl}}$ parallel Gaussian channels, denoted by $C^{\rm{dl}} (t)$, is the solution of the following optimization problem \cite{LiangVinceSecure,NirScheersDisOnlyCommInf}:
\begin{align}\label{D_CapacityDef}
&\mathop {\max }\limits_{P_{1}, \dots, P_{n^{\rm{dl}}}} \mathop {\min }\limits_{m \in [M]} \sum\nolimits_{i=1}^{n^{\rm{dl}}} {\log _2}\left( 1 + P_{m,i}^{\rm{dl}} (t) \left| {h}^{\rm{dl}}_{m,i} \left(t\right) \right|^2 \right), \nonumber\\
& \mbox{subject to $\sum\nolimits_{i=1}^{n^{\rm{dl}}} P_{m,i}^{\rm{dl}} (t) = {P}^{\rm{dl}}$, $\forall m \in [M]$}.
\end{align}
The above problem is a convex optimization problem which can be efficiently solved by the minimax hypothesis testing approach \cite{LiangVinceSecure,LiangVinceResMaxMin,NirScheersDisOnlyCommInf}. 
Note that this rate would be achievable by coding across infinitely many realizations of the $n^{\rm{dl}}$ parallel Gaussian channels under consideration, and will serve as an upper bound on the rate transmitted over a single realization. 
%It is evident that the highest $C^{\rm{dl}} (t)$ is achieved by setting
%\begin{align}\label{D_M_t}
%\mathcal{M} (t) = {\max}_{[K]} \left\{ \left| h^{\rm{dl}}_1 (t) \right|, ..., \left| h^{\rm{dl}}_M (t) \right|  \right\}. 
%\end{align}

\subsection{Compression Technique}\label{SubSecDownCompression} 
In the following, we present the compression technique employed by the PS for transmitting information about the global model over the bandwidth-limited downlink channel, where we adopt the scheme introduced in \cite{DCAlistarhQSGD} with a slight modification. 
Assume that vector $\boldsymbol{x} (t) \in \mathbb{R}^d$, whose $i$-th entry is denoted by ${x}_i (t)$, $i \in [d]$, is to be quantized and transmitted over the downlink channel by the PS.
The PS first sparsifies $\boldsymbol{x} (t)$ by setting all but $s$ entries of $\boldsymbol{x} (t)$ with the highest magnitudes to zero, for some integer $s \le {d}$. 
We denote the set of $s$ indices of the resultant sparse vector with non-zero entries by $\mathcal{S} (t)$.  
We also denote the resultant vector with dimension $s$ after removing the zeroed entries due to the sparsification by $\boldsymbol{x}_s (t)$, whose $i$-th entry is denoted by ${x}_{s, i} (t)$, for $i \in [s]$.
%, and denote the resultant sparse vector by $\widetilde{\boldsymbol{x}} (t)$.
%, and denote the resultant sparse vector by $\widetilde{\boldsymbol{x}} (t)$.
Then the PS quantizes the entries of $\boldsymbol{x}_s (t)$, and transmits the quantized values along with their locations in $\boldsymbol{x} (t)$, which are available in set $\mathcal{S} (t)$. 
We define
\begin{subequations}
\begin{align}\label{DefThetaMaxMin}
{x}_{\rm{max}} &\triangleq \max_{i \in \left[s \right]} \left\{ \left| {{x}_{s,i}} (t) \right| \right\},\\
{x}_{\rm{min}} &\triangleq \min_{i \in \left[s \right]} \left\{ \left| {{x}_{s,i}} (t) \right| \right\}.
\end{align}
\end{subequations}
%, and $\mathcal{Q} (\boldsymbol{v}(t))$ represents the compressed version of vector $\boldsymbol{v}(t)$. 
Given a quantization level $q(t)$, which will be determined later, we define the compression technique applied to the $i$-th entry of ${\boldsymbol{x}}_s (t)$, for $i \in [s]$, as 
\begin{subequations}\label{DefQuanTech}
\begin{align}
{Q} \left({x}_{s,i} (t)\right) \triangleq {\rm{sign}} \left( {x}_{s,i} (t) \right) \cdot \Big( x_{\rm{min}} + \left(x_{\rm{max}} - x_{\rm{min}} \right) \cdot  \varphi \Big( \frac{| {x}_{s,i} (t) | - x_{\rm{min}}}{x_{\rm{max}} -x_{\rm{min}}}, q(t) \Big) \Big),    
\end{align}
where, for $x \in \mathbb{R}$,  
\begin{align}
{\rm{sign}} \left( x \right) \triangleq \begin{cases} 
1, & \mbox{if $x \ge 0$},\\
-1, & \mbox{otherwise},
\end{cases}    
\end{align}
and $\varphi(\cdot, \cdot)$ is a quantization function defined in the following. For $0 \le x \le 1$ and some integer $q \ge 1$, let $l \in \{ 0, 1, \dots, q-1 \}$ be an integer such that $x \in [l/q, (l+1)/q)$. We then define
\begin{align}\label{DefQ}
\varphi \left( x, q \right) \triangleq \begin{cases} 
l / q, & \mbox{with probability $1 - \left( x q - l\right)$},\\
(l+1) / q, & \mbox{with probability $x q - l$}.
\end{cases}
\end{align}
\end{subequations}
We denote the compressed version of ${x}_{i} (t)$ by $S \left({x}_{i}(t)\right)$, for $i \in [d]$, which is given by
\begin{align}
S \left({x}_i(t)\right) = 
\begin{cases} 
Q \left({x}_i(t)\right), & \mbox{if $i \in \mathcal{S} (t)$},\\
0, & \mbox{otherwise},
\end{cases}    
\end{align}
and represent $\boldsymbol{S} \left(\boldsymbol{x}(t)\right) = \left[S \left({x}_1 (t) \right), \dots, S \big({x}_{d} (t) \big)\right]^T$. 
Note that we normalize the entries of $\boldsymbol{x}_s (t)$ with $x_{\rm{max}} - x_{\rm{min}}$ rather than $\| \boldsymbol{x}_s (t) \|_2$ as introduced in \cite{DCAlistarhQSGD}.

\iffalse
\begin{lemma}\label{LemDigQuanVar}
The compression technique $Q(\boldsymbol{x}^{\rm{dl}} (t))$, introduced in \eqref{DefQuanTech}, provides an unbiased estimate of $\boldsymbol{x}^{\rm{dl}} (t)$, i.e., $\mathbb{E}_{\varphi} \left[  Q(\boldsymbol{x}^{\rm{dl}} (t)) \right] = \boldsymbol{x}^{\rm{dl}} (t)$,
%\begin{align}
%\mathbb{E}_Q \left[  Q(\boldsymbol{v} (t)) \right] = \boldsymbol{v} (t),
%\end{align}
where $\mathbb{E}_{\varphi}$ represents the expectation with respect to the randomness of the quantization technique. 
\end{lemma}
\begin{proof}
We first show that $\mathbb{E}_{\varphi} \left[ {\varphi}(x, q) \right] = x$. According to \eqref{DefQ}, we have
\begin{align}\label{UnbiasedQ}
\mathbb{E}_{\varphi} \left[ {\varphi}(x, q) \right] & = \Big(\frac{l}{q} \Big) \left( 1 + l - xq \right) + \Big(\frac{l+1}{q} \Big) \left( xq - l \right) = x.
\end{align}
From \eqref{DefQuanTech} and \eqref{UnbiasedQ}, we have 
\begin{align}
\mathbb{E}_{\varphi} \left[  {x}^{\rm{dl}}_i (t) \right] &= \left( {x}^{\rm{dl}}_{\rm{max}} (t) - {x}^{\rm{dl}}_{\rm{min}} (t) \right) \cdot {\rm{sign}} \left( {{x}^{\rm{dl}}_i} (t) \right) \cdot \mathbb{E}_{\varphi} \Big[ {\varphi} \Big( \frac{\left| {x}^{\rm{dl}}_i (t) \right|}{{x}^{\rm{dl}}_{\rm{max}} (t) - {x}^{\rm{dl}}_{\rm{min}} (t)}, q(t) \Big) \Big]\nonumber\\
& = {\rm{sign}} \left( {{x}^{\rm{dl}}_i} (t) \right) \cdot \left| {x}^{\rm{dl}}_i (t) \right| = {x}^{\rm{dl}}_i (t), \quad i \in [d].   
\end{align}
This completes the proof of Lemma \ref{LemDigQuanVar}.
\end{proof}
\fi

With the above compression technique, the PS needs to transmit
\begin{align}\label{Rdltbits}
R^{\rm{dl}}(t) = 64 + s \left( 1 + \log_2(q(t)+1) \right) + \log_2 \binom{ d}{s} \mbox{ bits}     
\end{align}
over the wireless broadcast channel to each of the devices, where 64 bits are used to represent the real numbers ${x}_{\rm{max}}$ and ${x}_{\rm{min}}$, $s$ bits for presenting ${\rm{sign}} \left( {{x}_{s,i}} (t) \right)$, $\forall i \in [s]$, $s \log_2(q(t)+1)$ bits are used for $\varphi \left( \left(| {x}_{s,i} (t) |-x_{\rm{min}} \right)/ (x_{\rm{max}} - x_{\rm{min}}), q \right)$, $\forall i \in [s]$, and $\log_2 \binom{ d}{s}$ bits represent the indices of $\boldsymbol{x} (t)$ in set $\mathcal{S} (t)$. 
We set $q(t)$ to the largest integer satisfying $R^{\rm{dl}}(t) \le C^{\rm{dl}}(t)$.

\subsection{Model Update}\label{SubSecModelUpdateDig}
Here we present the model update scheme including the global model update broadcasting from the PS to the devices and aggregation of the local updates via uplink transmission from the devices to the PS.

\noindent \textbf{Downlink transmission.} We first elaborate on the downlink transmission.
We highlight that, for the digital downlink approach, all the devices have the same estimate of ${\boldsymbol{\theta}} (t)$ during global iteration $t$, denoted by $\wh{\boldsymbol{\theta}} (t)$, i.e., $\wh{\boldsymbol{\theta}}_m (t) = \wh{\boldsymbol{\theta}} (t)$, $\forall m \in [M]$. 
%This will become clear after presenting the downlink transmission.
In the downlink, at the global iteration $t$, the PS wants to broadcast the global model update $\boldsymbol{\theta} (t) - \wh{\boldsymbol{\theta}} (t-1)$ to all the devices.  
%over the wireless channel, where we will show that the server knows about $\wh{\boldsymbol{\theta}} (t-1)$; 
We define 
\begin{align}
\Delta \wh{\boldsymbol{\theta}} (t -1) \triangleq \boldsymbol{\theta} (t) - \wh{\boldsymbol{\theta}} (t-1) \in \mathbb{R}^d.    
\end{align}
The PS first quantizes $\Delta \wh{\boldsymbol{\theta}} (t -1)$ using the compression technique described in Section \ref{SubSecDownCompression}, obtaining $\boldsymbol{S} \left( \Delta \wh{\boldsymbol{\theta}} (t -1) \right) = \boldsymbol{S}\left(\boldsymbol{\theta} (t) - \wh{\boldsymbol{\theta}} (t-1) \right)$, which results in $R^{\rm{dl}}(t)$ bits as given in \eqref{Rdltbits}. 
The PS then broadcasts these bits to all the devices using a capacity achieving channel code, where $q(t)$ is set to the largest integer satisfying $R^{\rm{dl}}(t) \le C^{\rm{dl}}(t)$, where $C^{\rm{dl}}(t)$ given as the solution of \eqref{D_CapacityDef}. 
After decoding $\boldsymbol{S}\big(\boldsymbol{\theta} (t) - \wh{\boldsymbol{\theta}} (t-1) \big)$, each device computes $\wh{\boldsymbol{\theta}} (t)$ as
\begin{align}
\wh{\boldsymbol{\theta}} (t) = \wh{\boldsymbol{\theta}} (t-1) + \boldsymbol{S}\big(\boldsymbol{\theta} (t) - \wh{\boldsymbol{\theta}} (t-1) \big),    
\end{align}
which is equivalent to
\begin{align}
\wh{\boldsymbol{\theta}} (t) = {\boldsymbol{\theta}} (0) + \sum\nolimits_{i=1}^{t} \boldsymbol{S}\big(\boldsymbol{\theta} (i) - \wh{\boldsymbol{\theta}} (i-1) \big),    
\end{align}
where we have assumed that $\wh{\boldsymbol{\theta}} (0) = {\boldsymbol{\theta}} (0)$.
Having knowledge about the compressed vector $\boldsymbol{S}\big(\boldsymbol{\theta} (i) - \wh{\boldsymbol{\theta}} (i-1) \big)$, $\forall i \in [t]$, the PS can also recover $\wh{\boldsymbol{\theta}} (t)$, which is used at the devices to compute the local updates.

\noindent \textbf{Uplink transmission.} For ease of presentation, we assume that $n^{\rm{up}} = d/2$, and we will discuss the generalization of the prorposed approach.
%We schedule devices for uplink transmission based on their channel conditions, such that only devices with absolute value of the channel gains higher than a threshold $\gamma_{\rm{thr}}$ participate in the uplink transmission; that is, we set 
%\begin{align}\label{M_tUplink}
%\mathcal{M} (t) = \left\{ m: \left| h_m(t) \right| \ge \gamma_{\rm{thr}} \right\},    
%\end{align} 
%where each device  the CSI knowledge at the devices, each device can  
Device $m$, $m \in [M]$, performs $\tau$ local SGD steps, where the $i$-th step is given by 
\begin{align}\label{D_ithStepSGDDevicem}
\boldsymbol{\theta}_m^{i+1} (t) = \boldsymbol{\theta}_m^i (t) - \eta^i_m (t) \nabla F_m \left( \boldsymbol{\theta}_m^i (t), \xi_m^i (t) \right),  \quad \mbox{$i \in [\tau]$},   
\end{align}
where $\boldsymbol{\theta}_m^1 (t) = \wh{\boldsymbol{\theta}} (t)$.
It then transmits the local model update $\Delta \boldsymbol{\theta}_m (t) = \boldsymbol{\theta}_m^{\tau+1} (t) - \wh{\boldsymbol{\theta}} (t)$ in an analog (uncoded) fashion. We define 
\begin{subequations}
\begin{align}
\Delta \boldsymbol{\theta}_{m,\rm{re}} (t) &\triangleq [\Delta \theta_{m,1} (t), \dots, \Delta \theta_{m,d/2} (t)]^T, \\
\Delta \boldsymbol{\theta}_{m,\rm{im}} (t) &\triangleq [\Delta \theta_{m, d/2+1} (t), \dots, \Delta \theta_{m,d} (t)]^T,
\end{align}
\end{subequations}
where $\Delta \theta_{m,i} (t)$ denotes the $i$-th entry of $\Delta \boldsymbol{\theta}_{m} (t)$, for $i \in [d]$, $m \in [M]$, and we have $\Delta \boldsymbol{\theta}_m (t) = \left[{\Delta \boldsymbol{\theta}_{m,\rm{re}} (t)}^T, {\Delta \boldsymbol{\theta}_{m,\rm{im}} (t)}^T \right]^T$. 
Device $m$, $m \in [M]$, transmits 
\begin{align}
\boldsymbol{x}^{\rm{ul}}_m (t) = \boldsymbol{\alpha}^{\rm{ul}}_m(t) \circ  \left( \Delta \boldsymbol{\theta}_{m,\rm{re}} (t) + j \Delta \boldsymbol{\theta}_{m,\rm{im}} (t) \right),    
\end{align}
where $\boldsymbol{\alpha}^{\rm{ul}}_m(t) \in \mathbb{C}^{d/2}$ is the power allocation vector, whose $i$-th entry, $i \in [d/2]$, is set as
\begin{align}\label{UplinkPowerAllocVec}
{\alpha}^{\rm{ul}}_{m,i} (t) =
\begin{cases} 
\frac{\gamma_m (t)}{h^{\rm{ul}}_{m,i}(t)}, & \mbox{if $| h^{\rm{ul}}_{m,i}(t) | \ge \lambda_{{\rm{thr}}} (t)$},\\
0, &\mbox{otherwise}, 
\end{cases}    
\end{align}
for some $\gamma_m (t), \lambda_{{\rm{thr}}} (t) \in \mathbb{R}$, which are set to satisfy the transmit power constraint $\| \boldsymbol{x}^{\rm{ul}}_m (t) \|_2^2 \le {P}^{\rm{ul}}$. 
We assume that device $m$ first transmits the scaling factor $\gamma_m (t)$ to the PS in an error-free fashion, $m \in [M]$. 
%over the wireless MAC, where ${\alpha}^{\rm{ul}}_m(t)$ is set to satisfy the average power constraint at device $m$, and we assume that it is shared with the PS in an error-free fashion, $m \in \mathcal{M} (t)$. 
The PS receives the following signal:
\begin{align}
\boldsymbol{y}^{\rm{ul}} (t) = \sum\nolimits_{m = 1}^{M} \boldsymbol{\alpha}^{\rm{ul}}_m(t) \circ  \left( \Delta \boldsymbol{\theta}_{m,\rm{re}} (t) + j \Delta \boldsymbol{\theta}_{m,\rm{im}} (t) \right) \circ \boldsymbol{h}^{\rm{ul}}_{m}(t) +  \boldsymbol{z}^{\rm{ul}} (t),    
\end{align}
whose $i$-th entry, $i \in [d/2]$, is given by
\begin{align}
{y}^{\rm{ul}}_i (t) = \sum\nolimits_{m \in \mathcal{M}_i(t)} \gamma_m (t) \left( \Delta {\theta}_{m,i} (t) + j \Delta {\theta}_{m,d/2+i} (t) \right) + {z}^{\rm{ul}}_i (t),  
\end{align}
where we have defined 
\begin{align}
\mathcal{M}_i(t) \triangleq \left\{ m \in [M]: \left| h^{\rm{ul}}_{m,i}(t) \right| \ge \lambda_{{\rm{thr}}} (t) \right\}.    
\end{align}
\begin{algorithm}[t]
\caption{Digital Downlink Approach}
\label{ModelUpdateAlg}
\begin{algorithmic}[1]
\Statex
\State{\textbf{Initialize} $\boldsymbol{\theta} (0)$}
\For {$t = 0, \ldots, T-1$}
\Statex
\begin{itemize}
\item \textbf{Downlink transmission:}
\end{itemize}
%\State{$\boldsymbol{x}^{\rm{dl}} (t) = \boldsymbol{\theta} (t) - \wh{\boldsymbol{\theta}} (t-1) = \Delta \wh{\boldsymbol{\theta}} (t -1)$}
\State{PS broadcasts $\boldsymbol{S}\big(\boldsymbol{\theta} (t) - \wh{\boldsymbol{\theta}} (t-1) \big)$}
\State{$\wh{\boldsymbol{\theta}} (t) = \wh{\boldsymbol{\theta}} (t-1) + \boldsymbol{S}\big(\boldsymbol{\theta} (t) - \wh{\boldsymbol{\theta}} (t-1) \big)$}
\Statex
\begin{itemize}
\item \textbf{Uplink transmission:}
\end{itemize}
\For {$m = 1, \ldots, M$ in parallel}
\State{$\boldsymbol{x}^{\rm{ul}}_m (t) = \boldsymbol{\alpha}^{\rm{ul}}_m(t) \circ  \left( \Delta \boldsymbol{\theta}_{m,\rm{re}} (t) + j \Delta \boldsymbol{\theta}_{m,\rm{im}} (t) \right)$}
\State{${\alpha}^{\rm{ul}}_{m,i} (t) =
\begin{cases} 
\frac{\gamma_m (t)}{h^{\rm{ul}}_{m,i}(t)}, & \mbox{if $\left| h^{\rm{ul}}_{m,i}(t) \right| \ge \lambda_{{\rm{thr}}} (t)$},\\
0, &\mbox{otherwise} 
\end{cases}$, \; for $i \in [d/2]$}
\EndFor
\State{$\boldsymbol{\theta} (t+1) = \wh{\boldsymbol{\theta}} (t) + \Delta \wh{\boldsymbol{\theta}} (t)$}
\EndFor
\end{algorithmic}\label{Dig_Downlink}
\end{algorithm}
\hspace{-.27cm} With the knowledge of the channel state, and consequently $\mathcal{M}_i(t)$, $\forall i \in [d/2]$, the PS estimates $\frac{1}{\left| \mathcal{M}_i(t) \right|} \sum\nolimits_{m \in \mathcal{M}_i(t)} \Delta {\theta}_{m,i} (t)$ and $\frac{1}{\left| \mathcal{M}_{i}(t) \right|} \sum\nolimits_{m \in \mathcal{M}_i(t)} \Delta {\theta}_{m,d/2+i} (t)$ with
\begin{subequations}\label{DigitalUplinkDecoder}
\begin{align}
\Delta \hat{\theta}_{i} (t) & =  
\begin{cases}
\frac{{\rm{Re}} \left\{ {y}^{\rm{ul}}_i (t) \right\}}{\bar{\gamma} (t) \left| \mathcal{M}_i(t) \right|}, & \mbox{if $\left| \mathcal{M}_i(t) \right| \ne 0$},\\
0, & \mbox{otherwise},
\end{cases}\\
\Delta \hat{\theta}_{d/2+i} (t) & =  
\begin{cases}
\frac{{\rm{Im}} \left\{ {y}^{\rm{ul}}_i (t) \right\}}{\bar{\gamma} (t) \left| \mathcal{M}_i(t) \right|}, & \mbox{if $\left| \mathcal{M}_i(t) \right| \ne 0$},\\
0, & \mbox{otherwise},
\end{cases}
\end{align}
\end{subequations}
respectively, where we have defined $\bar{\gamma} (t) \triangleq \frac{1}{M} \sum\nolimits_{m =1 }^{M} {\gamma}_m (t)$.
The estimated vector $\Delta \wh{\boldsymbol{\theta}} (t) \triangleq [ \Delta \hat{\theta}_{1} (t) , \dots, \Delta \hat{\theta}_{d} (t)]^T$ is used to update the global model parameter vector as 
\begin{align}
\boldsymbol{\theta} (t+1) = \wh{\boldsymbol{\theta}} (t) + \Delta \wh{\boldsymbol{\theta}} (t).    
\end{align}
We remark here that for $n^{\rm{up}} < d/2$, we carry out the uplink transmission in $\left\lceil {d/(2n^{\rm{up}})} \right\rceil$ time slots, where in each time slot we perform the above transmission.

Algorithm \ref{Dig_Downlink} summarizes the downlink and uplink transmissions for the digital downlink approach employing the compression technique presented in Section \ref{SubSecDownCompression}.

\section{Analog Downlink Approach}\label{SecAnalog}

In this section, we propose that the PS broadcasts the global model parameter vector $\boldsymbol{\theta}(t)$ in an analog (uncoded) manner. For ease of presentation, we consider $n^{\rm{dl}} = d/2$,
%, and drop the dependence of all the setting parameters on $n^{\rm{dl}}$, 
and we will argue that the proposed approach can be readily extended to the general case.

%The PS aims to schedule devices with the best estimates of the global model parameter vector for participation in the uplink transmission. To this end, having downlink CSI knowledge at the PS, similarly to the digital downlink approach, we consider the downlink channel conditions as the device scheduling metric, and set
%\begin{align}\label{A_M_t}
%\mathcal{M} (t) = {\max}_{[K]} \left\{ \left| h^{\rm{dl}}_1 (t) \right|, ..., \left| h^{\rm{dl}}_M (t) \right|  \right\}.   
%\end{align}
%This device scheduling technique relies on the fact that the devices with better channel gains can have a better estimation of $\boldsymbol{\theta}(t)$.  

\noindent \textbf{Downlink transmission.} We define
\begin{subequations}
\begin{align}
\boldsymbol{\theta}_{\rm{re}} (t) &\triangleq [\theta_1 (t), \dots, \theta_{d/2} (t)]^T, \\
\boldsymbol{\theta}_{\rm{im}} (t) &\triangleq [\theta_{d/2+1} (t), \dots, \theta_{d} (t)]^T,
\end{align}
\end{subequations}
where $\boldsymbol{\theta} (t) = \left[ {\boldsymbol{\theta}_{\rm{re}} (t)}^T, {\boldsymbol{\theta}_{\rm{im}} (t)}^T \right]^T$.
At the global iteration $t$, the PS broadcasts $\boldsymbol{x}^{\rm{dl}} (t) = {\alpha}^{\rm{dl}} (t) \left(\boldsymbol{\theta}_{\rm{re}} (t) + j \boldsymbol{\theta}_{\rm{im}} (t) \right)$ in an uncoded manner, where ${\alpha}^{\rm{dl}} (t)$ is set to satisfy $\| \boldsymbol{x}^{\rm{dl}} (t) \|_2^2 \le {P}^{\rm{dl}}$. Before broadcasting $\boldsymbol{x}^{\rm{dl}} (t)$, we assume that the PS shares ${\alpha}^{\rm{dl}} (t)$ with the devices in an error-free fashion. The received signal at device $m$ is given by 
\begin{align}
\boldsymbol{y}_{m}^{\rm{dl}} (t) = {\alpha}^{\rm{dl}} (t) \boldsymbol{h}^{\rm{dl}}_m (t) \circ \left(\boldsymbol{\theta}_{\rm{re}} (t) + j \boldsymbol{\theta}_{\rm{im}} (t) \right) +  \boldsymbol{z}_m^{\rm{dl}} (t), \quad m \in [M].    
\end{align}
Device $m$, $m \in [M]$, performs the following descaling:
\begin{align}\label{YHat_A_Downlink}
\wh{\boldsymbol{y}}_{m}^{\rm{dl}} (t) \triangleq \Big( \frac{1}{{\alpha}^{\rm{dl}} (t)} \Big) \boldsymbol{y}_{m}^{\rm{dl}} (t) \circ \left( \boldsymbol{h}^{\rm{dl}}_m (t) \right)^{-1} = \boldsymbol{\theta}_{\rm{re}} (t) + j \boldsymbol{\theta}_{\rm{im}} (t) + \Big( \frac{1}{{\alpha}^{\rm{dl}} (t)} \Big) \boldsymbol{z}_{m}^{\rm{dl}} (t) \circ \left( \boldsymbol{h}^{\rm{dl}}_m (t) \right)^{-1},    
\end{align}
and uses $\wh{\boldsymbol{y}}_{m}^{\rm{dl}} (t)$ to recover the global model parameter vector $\boldsymbol{\theta} (t)$ as
\begin{align}\label{A_RecoveredThetaDL}
\wh{\boldsymbol{\theta}}_m (t) \triangleq \left[{\rm{Re}} \{ \wh{\boldsymbol{y}}_{m}^{\rm{dl}} (t) \}^T , {\rm{Im}} \{ \wh{\boldsymbol{y}}_{m}^{\rm{dl}} (t) \}^T \right]^T.
\end{align}
We highlight that the proposed approach can be extended for any number of subchannels $n^{\rm{dl}}$ through transmission over different time slots.

\noindent \textbf{Uplink transmission.} After recovering $\wh{\boldsymbol{\theta}}_m (t)$, device $m$, $ m \in [M]$, performs $\tau$ local SGD steps as in \eqref{D_ithStepSGDDevicem}, where $\boldsymbol{\theta}_m^1 (t) = \wh{\boldsymbol{\theta}}_m (t)$. 
It then transmits the local model update $\Delta \boldsymbol{\theta}_m (t) = \boldsymbol{\theta}_m^{\tau+1} (t) - \wh{\boldsymbol{\theta}}_m (t)$ in an analog (uncoded) fashion over the wireless MAC, $m \in [M]$. 
The uplink transmission follows the same steps as the one presented in Section \ref{SubSecModelUpdateDig} for the digital downlink approach. However, the PS recovers $\Delta \wh{\boldsymbol{\theta}} (t)$, given in \eqref{DigitalUplinkDecoder}, and updates the global model parameter vector as $\boldsymbol{\theta} (t+1) = \boldsymbol{\theta} (t) + \Delta \wh{\boldsymbol{\theta}} (t)$.

\begin{remark}\label{RemAnalogPrivacy}
We highlight that with the independent random noise added to the model parameter vector in the downlink at different devices, the analog downlink approach inherently introduces additional data privacy for the FL framework.  
%\cite{JunVincePrivacyFL,SeifTandonLocalPrivacyFL}.    
\end{remark}

\section{Convergence Analysis of Analog Downlink Approach}\label{SecConvergence}
Here we analyze convergence behavior of the analog downlink approach presented in Section \ref{SecAnalog}. For simplicity of the convergence analysis, we assume that the device-to-PS transmission is error-free, and focus on the impact of noisy downlink transmission on the convergence performance.
%\footnote{The results here can be extended to the noisy uplink case through a similar procedure introduced in \cite{}, which provides a convergence result for the noisy uplink scenario.}. 
%We show that the gap between the expected loss function and the minimum loss function approaches zero for large enough $T$. 
We first present the preliminaries and assumptions, and then the convergence result for the analog downlink approach, whose proof is provided in the Appendix.
%We first establish convergence results for the full device participation scenario, in which $K=M$, and then extend the results to the partial device participation case, in which $K < M$. 

\subsection{Preliminaries}\label{SubSecPreConvergence}
We define the optimal solution of minimizing $F \left( \boldsymbol{\theta} \right)$ as 
\begin{align}\label{ThetaStarDef}
\boldsymbol{\theta}^* \triangleq \arg \mathop {\min }\limits_{\boldsymbol{\theta}} F(\boldsymbol{\theta}),  
\end{align}
and the minimum loss as $F^* \triangleq F(\boldsymbol{\theta}^*)$. We also denote the minimum value of $F_m (\cdot)$, the local loss function at device $m$, by $F_m^*$, $m \in [M]$. We then define
\begin{align}\label{GammaDef}
\Gamma \triangleq F^* -  \sum\nolimits_{m=1}^{M} \frac{B_m}{B} F^*_m,    
\end{align}
where $\Gamma \ge 0$, and its magnitude indicates the bias in the data distribution across devices. We note that for i.i.d. data distribution, given a large enough number of local data samples, $\Gamma$ approaches zero.

According to \eqref{YHat_A_Downlink} and \eqref{A_RecoveredThetaDL}, we have
\begin{align}\label{ThetamHatAnalogDownlinkConvergence}
\wh{\boldsymbol{\theta}}_m (t) = {\boldsymbol{\theta}} (t) + \widetilde{\boldsymbol{z}}_m^{\rm{dl}} (t), \quad \mbox{for $m \in [M]$},    
\end{align}
where, for ease of presentation, we have defined 
\begin{align}
\widetilde{\boldsymbol{z}}_m^{\rm{dl}} (t) \triangleq \Big( \frac{1}{{\alpha}^{\rm{dl}} (t)} \Big) \left[ {\rm{Re}} \big\{ \boldsymbol{z}_{m}^{\rm{dl}} (t) \circ \left( \boldsymbol{h}^{\rm{dl}}_m (t) \right)^{-1} \big\}^T , {\rm{Im}} \big\{ \boldsymbol{z}_{m}^{\rm{dl}} (t) \circ \left( \boldsymbol{h}^{\rm{dl}}_m (t) \right)^{-1} \big\}^T \right]^T.    
\end{align}
For simplicity of the convergence analysis, we consider $\eta_m^i (t) = \eta(t)$, $\forall m ,i$. Thus, the $i$-th step local SGD at device $m$ is given by
\begin{align}\label{ConvSGDDevicem}
\boldsymbol{\theta}_m^{i+1} (t) = \boldsymbol{\theta}_m^i (t) - \eta (t) \nabla F_m \left( \boldsymbol{\theta}_m^i (t), \xi_m^i (t) \right),  \quad \mbox{$i \in [\tau]$}, \mbox{$m \in [M]$},   
\end{align}
where $\boldsymbol{\theta}_m^1 (t) = \wh{\boldsymbol{\theta}}_m (t)$, given in \eqref{ThetamHatAnalogDownlinkConvergence}. Thus, we have
\begin{align}\label{ConvSGDDevicemith_2}
\boldsymbol{\theta}_m^{\tau+1} (t) = {\boldsymbol{\theta}}_m^1 (t) - \eta (t) \sum\nolimits_{i=1}^{\tau} \nabla F_m \left( \boldsymbol{\theta}_m^i (t), \xi_m^i (t) \right), \quad \mbox{for $m \in [M]$}.    
\end{align}
Device $m$ transmits the local model update $\Delta \boldsymbol{\theta}_m (t) = - \eta (t) \sum\nolimits_{i=1}^{\tau} \nabla F_m \left( \boldsymbol{\theta}_m^i (t), \xi_m^i (t) \right)$, $m \in [M]$. After receiving the local model updates from all the devices, $\Delta \boldsymbol{\theta}_m (t)$, $\forall m \in [M]$, the PS updates the global model parameter vector as
\begin{align}\label{D_ConvPSGlobalUpdateModelDifScheduled}
\boldsymbol{\theta} (t+1) = {\boldsymbol{\theta}} (t) +  \sum\limits_{m = 1}^M \frac{B_m}{B} \Delta \boldsymbol{\theta}_m (t) = {\boldsymbol{\theta}} (t) - \eta (t) \sum\limits_{m = 1}^M \sum\limits_{i=1}^{\tau} \frac{B_m}{B} \nabla F_m \left( \boldsymbol{\theta}_m^i (t), \xi_m^i (t) \right).
\end{align}

\begin{assumption}\label{AssumpSmoothLoss}
The loss functions $F_1, \dots, F_M$ are all $L$-smooth; that is, $\forall \boldsymbol{v}, \boldsymbol{w} \in \mathbb{R}^d$, 
\begin{align}\label{ConvLSmoothCondit}
F_m(\boldsymbol{v}) - F_m(\boldsymbol{w}) \le \langle \boldsymbol{v} - \boldsymbol{w} , \nabla F_m (\boldsymbol{w}) \rangle + \frac{L}{2} \left\| \boldsymbol{v} - \boldsymbol{w} \right\|^2_2, \quad \forall m \in [M].
\end{align}
\end{assumption}

\begin{assumption}\label{AssumpStrongConvexLoss}
The loss functions $F_1, \dots, F_M$ are all $\mu$-strongly convex; that is, $\forall \boldsymbol{v}, \boldsymbol{w} \in \mathbb{R}^d$, 
\begin{align}\label{ConvMuStConvexCondit}
F_m(\boldsymbol{v}) - F_m(\boldsymbol{w}) \ge \langle \boldsymbol{v} - \boldsymbol{w} , \nabla F_m (\boldsymbol{w}) \rangle + \frac{\mu}{2} \left\| \boldsymbol{v} - \boldsymbol{w} \right\|^2_2, \quad \forall m \in [M].      
\end{align}
\end{assumption}

\begin{assumption}\label{AssumpBoundedVarGradient}
The expectation of the squared $l_2$-norm of the stochastic gradients are bounded; that is,
\begin{align}\label{ConvNorm2Bound}
\mathbb{E}_{\xi} \left [ \left\| \nabla F_m \left( \boldsymbol{\theta}_m^i (t), \xi_m^i (t) \right) \right\|^2_2 \right] \le G^2, \quad \forall i \in [\tau], \forall m \in [M], \; \forall t.      
\end{align}
\end{assumption}

\begin{assumption}\label{AssumpBoundThetaHatmAnalog}
We assume 
\begin{align}\label{Assumption_Analog_Conv}
\mathbb{E} \left[ \left\| \sum\nolimits_{m=1}^{M} \frac{B_m}{B} \left( \nabla F_m ( \boldsymbol{\theta} (t) + \widetilde{\boldsymbol{z}}_m^{\rm{dl}} (t), \xi_m^1 (t) ) - \nabla F_m ( \boldsymbol{\theta} (t), \xi_m^1 (t) ) \right) \right\|^2 \right] \le \frac{ Z^2}{M \sigma^{\rm{dl}} {P}^{\rm{dl}}},    
\end{align}
for some $Z \in \mathbb{R}$, where the upper bound reduces with the variance of the downlink channel gains, the downlink transmit power, and the number of devices, $M$. We have assumed that the effect of the downlink noise is alleviated by averaging over the devices. 
\end{assumption}

\subsection{Convergence Rate}\label{SubSecConvDig}
Here we provide the convergence rate for the analog downlink approach introduced in Section \ref{SecAnalog} assuming that the devices can send their local model updates accurately.

\begin{theorem}\label{A_Theoremtheta_thetastarKequalM}
Let $0 < \eta(t) \le \min \left\{ \frac{\mu}{\mu + 1}, \frac{1}{\mu \tau} \right\}$, $\forall t$. For the analog downlink approach, we have
\begin{subequations}\label{A_ConvTheoremtheta_thetastarKequalM}
\begin{align}\label{A_ConvTheoremtheta_thetastarKequalM_main}
\mathbb{E} \left[ \left\| \boldsymbol{\theta} (t) - {\boldsymbol{\theta}}^* \right\|_2^2 \right] \le  \left( \prod\nolimits_{i=0}^{t-1} A(i) \right) \left\| {\boldsymbol{\theta}} (0) - {\boldsymbol{\theta}}^* \right\|_2^2 + \sum\nolimits_{j=0}^{t-1} B(j) \prod\nolimits_{i=j+1}^{t-1} A(i),  
\end{align}
where 
\begin{align}\label{A_ConvTheoremtheta_thetastarKequalM_A}
A(i) \triangleq & 1 - \mu \eta (i) \left( \tau - \eta(i) (\tau - 1 + 1/\mu) \right),\\ 
B(i) \triangleq &  \frac{Z^2}{M \sigma^{\rm{dl}} {P}^{\rm{dl}}} + \left( 1+ \mu (1- \eta(i)) \right) \eta^2(i) G^2 \frac{\tau (\tau-1)(2\tau-1)}{6} \nonumber\\
&+  \left( \tau - 1 +\eta^2(i) (\tau^2 + \tau-1) \right) G^2  + 2  \eta(i) (\tau - 1) \Gamma, \label{A_ConvTheoremtheta_thetastarKequalM_B}
\end{align}
\end{subequations}
and the expectation is with respect to the stochastic gradient function and the randomness of the underlying wireless channel.
\end{theorem}
\begin{proof}
See Appendix \ref{A_AppTheorem}. 
\end{proof}

\begin{corollary}\label{A_CorrConvF_FstarKequalM}
From the $L$-smoothness of function $F(\cdot)$, after $T$ global iterations of the analog downlink scheme, for $0 < \eta(t) \le \min \left\{ \frac{\mu}{\mu + 1}, \frac{1}{\mu \tau} \right\}$, $\forall t$, we have
\begin{align}\label{A_ConvF_FstarKequalM}
\mathbb{E} \left[ F( \boldsymbol{\theta} (T)) \right] - F^* \le & \frac{L}{2} \mathbb{E} \left[ \left\| \boldsymbol{\theta} (T) - {\boldsymbol{\theta}}^* \right\|_2^2 \right] \nonumber \\
\le & \frac{L}{2} \left( \prod\nolimits_{i=0}^{T-1} A(i) \right) \left\| {\boldsymbol{\theta}} (0) - {\boldsymbol{\theta}}^* \right\|_2^2 + \frac{L}{2} \sum\nolimits_{j=0}^{T-1} B(j) \prod\nolimits_{i=j+1}^{T-1} A (i), 
\end{align}
where the last inequality follows from \eqref{A_ConvTheoremtheta_thetastarKequalM_main}. 
%Having a decreasing learning rate $\mathop {\lim }\limits_{t \to \infty} \eta(t) = 0$, it is easy to verify that $\mathop {\lim }\limits_{T \to \infty} \mathbb{E} \left[ F( \boldsymbol{\theta} (T) ) \right] - F^* = 0$.
\end{corollary}

\iffalse
\begin{remark}\label{RemConvKequalMpi}
For a limited $T$ value, it is hard to clearly discuss about the impact of $\rho(i)$ on the final average loss, since it depends on the other setting parameters, such as $\mu$, $L$, $M$, $K$, $\Gamma$ and $G$. However, since $A(i)$ reduces with $\rho(i)$, one can observe that the convergence speed increases with $\rho(i)$. 
%On the other hand, $B(i)$ increases with $\rho(i)$, which indicates that the largest value of $\rho(i)$, which yields by scheduling a single device, might not provide the best convergence result.  
\end{remark}
\fi

\begin{remark}\label{RemConvKlessMpi}
We remark that $A(i)$ is a decreasing function of $\tau$, while $B(i)$ increases with $\tau$. Therefore, the impact of $\tau$ on the convergence performance in the general case is not evident, since it depends also on other parameters. 
However, for a more biased data distribution across devices, which results in a higher $\Gamma$ and $G$, the destructive effect of increasing $\tau$ on $B(i)$ is more significant, while the reduction in $A(i)$ is the same as having a less biased data distribution. We note that $A(i)$ is not a function of the data distribution; 
therefore, for a less diverse data distribution, designing an efficient $\tau$ is more critical. 
This corroborates our intuitive understanding of convergence in this problem, where for a more biased data distribution, increasing the number of local iterations excessively leads to a more divergent local updates with a less chance of convergence.         
%The first term in $B(i)$, which is due to device scheduling, is a decreasing function of $K$, where we note that $\rho(i)$ reduces with $K$ due to the resource sharing. Thus, $B(i)$ reduces with $K$, and $K = M$ minimizes $B(i)$. On the other hand, since $A(i)$ increases with $K$, which is due to the reduction in $\rho(i)$, the impact of $K$ on the convergence performance is complicated. We will observe in Fig. \ref{Fig_Convergence} that for different setting parameters different $K$ values provide the best convergence performance.     
%However, when the resources are abundant such that $\rho(t) = 1$, $\forall t$, it is trivially known that the full devices participation scenario, i.e., $K=M$, provides the best performance; this is corroborated by the result in Corollary \ref{CorrConvF_FstarKequalM}, where, for $\rho(t) = 1$, $\forall t$, $K=M$ gives the best performance.
%However, for the limited resources case studied in this paper, $\rho(i)$ reduces with $K$, and the dependency of the term on RHS of \eqref{ConvF_FstarKlessM} on $K$ is sophisticated.   
\end{remark}

\begin{remark}\label{RemTau}
The two terms, $\frac{Z^2}{M \sigma^{\rm{dl}} {P}^{\rm{dl}}}$ and $(\tau-1) G^2$ in $B(i)$, are not scaled with the learning rate, $\eta(i)$. 
Therefore, even for a decreasing learning rate, where $\mathop {\lim }\limits_{t \to \infty } \eta(t) = 0$, we have $\mathop {\lim }\limits_{t \to \infty } B(t) = \frac{Z^2}{M \sigma^{\rm{dl}} {P}^{\rm{dl}}} + (\tau-1) G^2 \ne 0$, which shows that $\mathop {\lim }\limits_{t \to \infty } \mathbb{E} \left[ F( \boldsymbol{\theta} (t)) \right] - F^* \ne 0$. 
We highlight that having these two terms is the result of the noisy downlink transmission, where $\frac{Z^2}{M \sigma^{\rm{dl}} {P}^{\rm{dl}}}$ and $(\tau-1) G^2$ have appeared in the convergence analysis in inequalities \eqref{A_AppLemmaTemr_2_Eq_5} and \eqref{AppLemmaTemr_1_D_Eq_12}, respectively, in the appendices.           
\end{remark}

\section{Numerical Experiments}\label{SecExperiments}

Here we compare the performance of the proposed digital and analog downlink approaches for image classification on the MNIST dataset \cite{LeCunMNIST} with $60000$ training and $10000$ test samples. We train a convolutional neural network (CNN) with 6 layers including two $5 \times 5$ convolutional layers with ReLU activation and the same padding, where the first and the second layers have 32 and 64 channels, respectively, each with stride 1, and followed by a $2 \times 2$ max pooling layer with stride 2. Also, the CNN has a fully connected layer with 1024 units and ReLU activation with dropout $0.8$ followed by a softmax output layer. We utilize ADAM optimizer \cite{ADAMDC} to train the CNN.

We consider two scenarios: in the \textit{iid data distribution} scenario, we randomly split the $60000$ training data samples to $M$ disjoint subsets, and allocate each subset of data samples to a different device; 
while in the \textit{non-iid data distribution} scenario, we split the training data samples with the same label (from the same class) to $M/5$ disjoint subsets (assume that $M$ is divisible by 5). 
We then assign two subsets of the data samples, each from a different label/class selected at random, to each device, such that each subset of the data samples is assigned to a single device.

%For the experiments, we consider $M=40$ devices in the system. 
We assume $n^{\rm{dl}}=n^{\rm{ul}}=d/2$ subchannels, and a variance of $\sigma^{\rm{dl}} = \sigma^{\rm{ul}} = 1$ for the downlink and uplink channel gains. 
We set the transmit power constraint at the devices to ${P}^{\rm{ul}} = 10$, and the threshold on the uplink channel gains to $\lambda_{{\rm{thr}}} (t) = 10^{-4}$, $\forall t$. 
We also set the sparsity level of the digital downlink approach to $s = \left\lfloor {d/50} \right\rfloor$ and the size of the local mini-batch sample for each local iteration to $\left| \xi_m^i (t) \right| = 500$, $\forall i, m, t$.   
%We set the number of local iterations at the devices to $\tau = 3$. 
We measure the performance as the accuracy with respect to the test samples, called \textit{test accuracy}, versus the global iteration count, $t$.    

For the analytical results on the convergence rate of the analog downlink approach, we set $\eta (t) = \frac{\min\left\{ \frac{\mu}{\mu+1}, \frac{1}{\mu \tau}\right\}}{(10^{-3}t + 1)}$, $\forall t$, and consider $M=40$ devices. We assume that $\mu = 0.2$, $L=10$, $\left\| {\boldsymbol{\theta}} (0) - {\boldsymbol{\theta}}^* \right\|_2^2 = 5 \times 10^3$, and $Z^2 = 2 \times 10^4$. We also model the iid and non-iid data distributions by setting $(G^2, \Gamma) = (10, 5)$ and $(G^2, \Gamma) = (100, 50)$, respectively, where we note that the non-iid scenario results in higher $G$ and $\Gamma$ values.

%We also consider $B=100$ and $B=1000$ local data samples for IID and nonIID data distribution scenarios, respectively

\begin{figure}[t!]
\centering
\begin{subfigure}{.5\textwidth}
  \centering
  \includegraphics[scale=0.55,trim={19pt 7pt 46pt 38pt},clip]{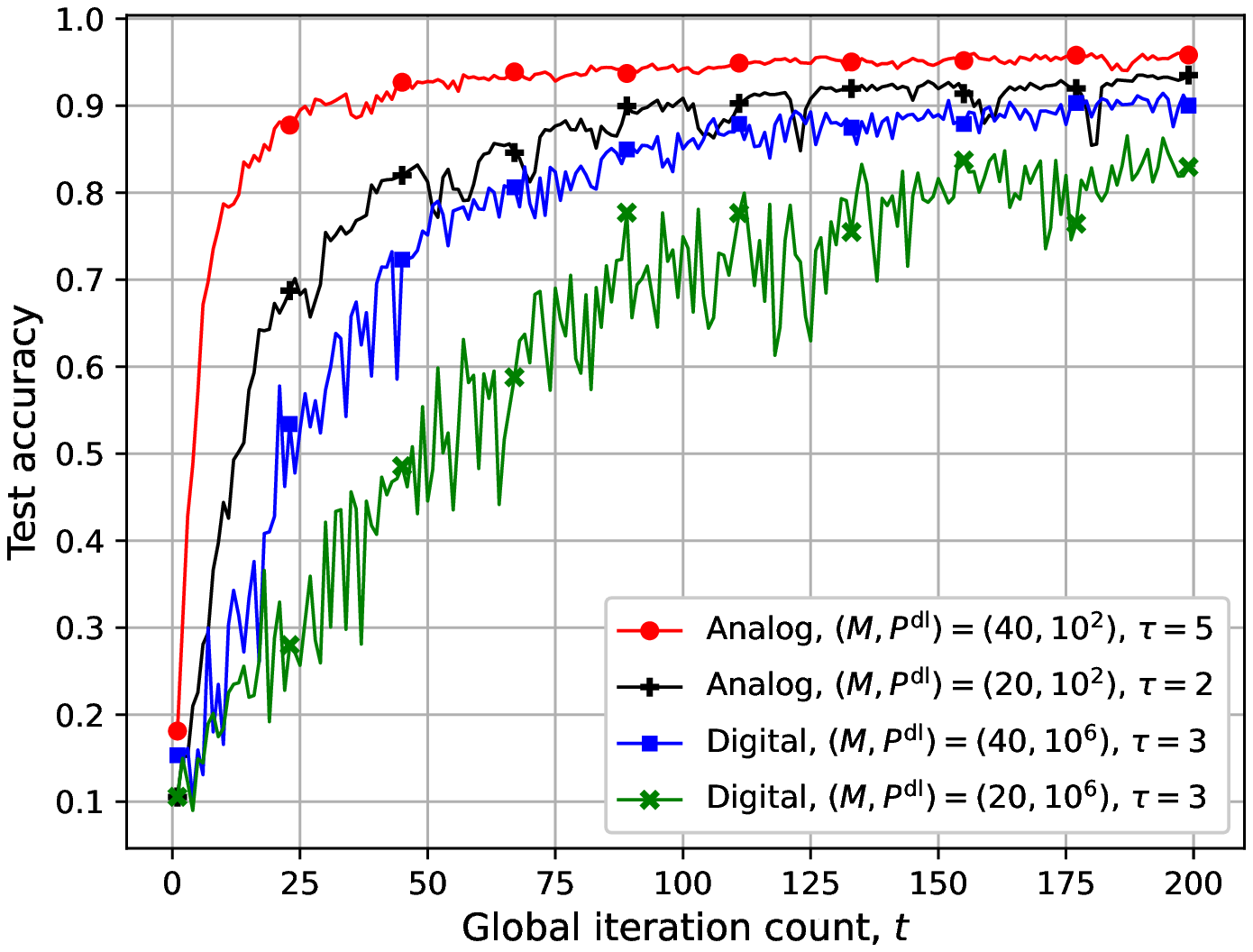}\vspace{0cm}
  \caption{Iid data distribution}
  \label{Fig_IID}
\end{subfigure}%
\begin{subfigure}{.5\textwidth}
  \centering
  \includegraphics[scale=0.55,trim={19pt 7pt 46pt 38pt},clip]{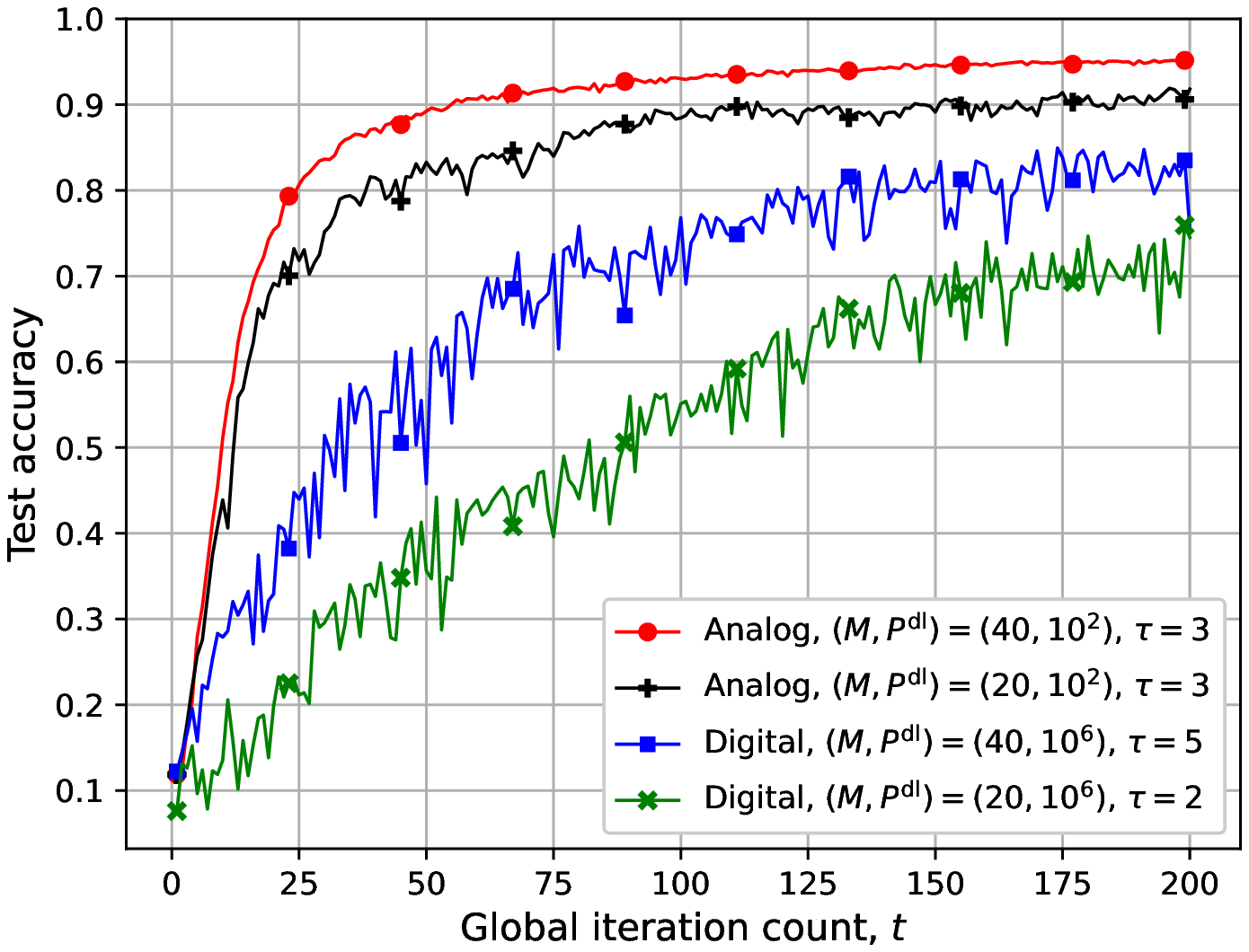}\vspace{0cm}
  \caption{Non-iid data distribution}
  \label{Fig_nonIID}
\end{subfigure}
\caption{Accuracy of the digital and analog downlink approaches for $n^{\rm{dl}}=n^{\rm{ul}}=d/2$, $\sigma^{\rm{dl}} = \sigma^{\rm{ul}} = 1$, $P^{\rm{ul}} = 10$, $\lambda_{{\rm{thr}}} (t) = 10^{-4}$, $\forall t$, $s = \left\lfloor {d/50} \right\rfloor$ for the digital approach, and $\left| \xi_m^i (t) \right| = 500$, $\forall i, m, t$.}
\label{Fig_IID_nonIID}
\end{figure}

In Fig. \ref{Fig_IID_nonIID} we compare the performance of the proposed digital and analog downlink approaches for both the iid and non-iid data distribution scenarios. We investigate the impact of the number of devices on the performance by considering $M \in \{ 20, 40\}$. 
For the analog downlink approach, we consider $P^{\rm{dl}} = 10^2$; while for the digital approach, we consider a significantly higher value for the downlink transmit power constraint at the PS, $P^{\rm{dl}} = 10^6$, which is to make sure that $q(t) \ge 1$, $\forall t$.
For each experiment, whose result is illustrated in Fig. \ref{Fig_IID_nonIID}, we have found the number of local iterations, $\tau$, which results in the best accuracy.
Despite the significantly lower transmit power at the PS, we observe that the analog downlink scheme remarkably outperforms the digital one for both iid and non-iid scenarios with a notably larger gap between the two for the non-iid case.  
It can also be seen that the accuracy of the analog downlink approach is more stable than its digital counterpart, and the degradation in the performance of the analog approach due to the introduced bias in the non-iid data distribution is marginal. 
This shows that the analog approach is fairly robust against the heterogeneity of data distribution across devices. 
We highlight that with the analog downlink approach the destructive effect of the devices with relatively bad channel conditions, and consequently with a noisier/less accurate estimate of the global model, is alleviated with the devices with good channel conditions, since devices receive different estimates of the global model vector transmitted by the PS depending on their channel conditions. 
On the other hand, with the digital downlink approach the common rate at which the global model vector is delivered to the devices should be adjusted such that all the devices, including those with relatively bad channel conditions, can decode it. 
This limits the capacity of the devices with good channel conditions, and provides the same copy of the global model estimate to all the devices whose rate is adjusted to accommodate even the worst device. 
Another reason for the inferiority of the digital downlink approach is that it requires digitization/quantization of the model parameter vector to a limited number of bits, which provides a less accurate estimate of the global model vector to rely on for local training at the devices than the noisy estimate received from the analog downlink transmission. This is due to the limited capacity of the wireless broadcast channel.

The performance of both digital and analog downlink approaches improve with $M$ for both iid and non-iid scenarios. 
This is mainly due to the uplink transmission. With more devices, each with its own power budget, analog transmission over the MAC is more robust against the noise, which is due to the additive nature of the MAC. 
However, the accuracy of the digital downlink approach is unstable in both iid and non-iid cases.
%for the smaller number of devices, $M=20$, in which the common rate is higher than the case when $M=40$, i.e., the devices have a better estimate of the model parameter vector. 
This is due to the inaccurate model parameter vector estimate at the devices for the digital downlink approach, which leads to a more skewed/less similar local updates at the devices compared to the case of having the actual model parameter vector at the devices. 
This deficiency can be clearly seen for $M=20$ in the iid scenario. By relying on the local updates from fewer devices, the chance of having more similar local updates (local updates with relatively small Euclidean distance) decreases, and it is less likely that the resultant vector recovered from the output of the MAC provides a good estimate of the gradient of the actual model parameter vector.    
%We further observe that the performance of the analog downlink approach is less vulnerable to the reduction in the transmit power at the PS than that of the digital approach.
Another interesting observation is about the best number of local iterations $\tau$ for each experiment. 
We observe that the best $\tau$ value for the analog downlink approach for $M=40$ ($M=20$) in the iid case is the same as that for the digital downlink approach for $M=40$ ($M=20$) in the non-iid scenario. 
\begin{figure}[t!]
\centering
\begin{subfigure}{.5\textwidth}
  \centering
  \includegraphics[scale=0.55,trim={19pt 7pt 46pt 38pt},clip]{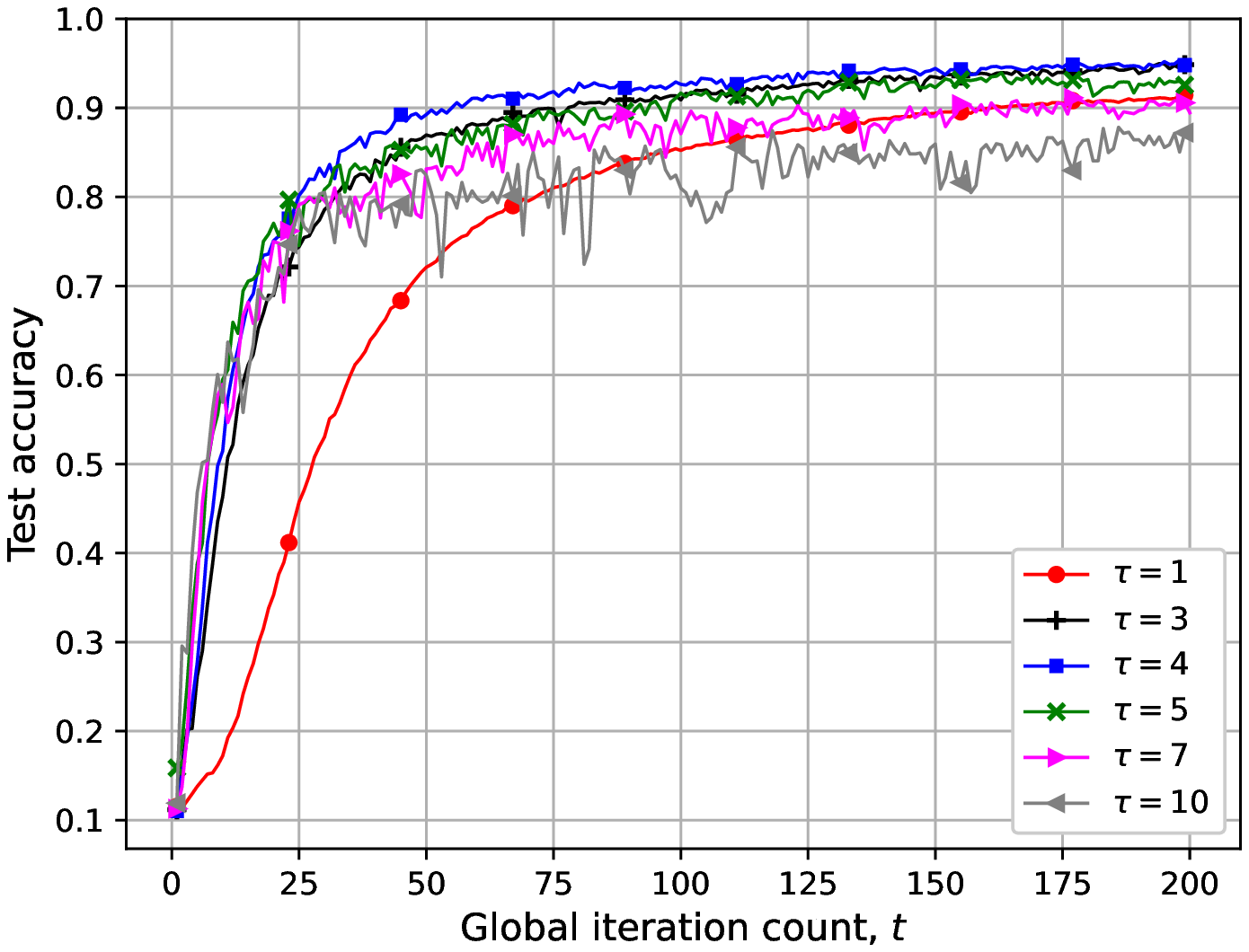}\vspace{0cm}
  \caption{$P^{\rm{dl}} = 10$}
  \label{Fig_Analog_low_power}
\end{subfigure}%
\begin{subfigure}{.5\textwidth}
  \centering
  \includegraphics[scale=0.55,trim={19pt 7pt 46pt 38pt},clip]{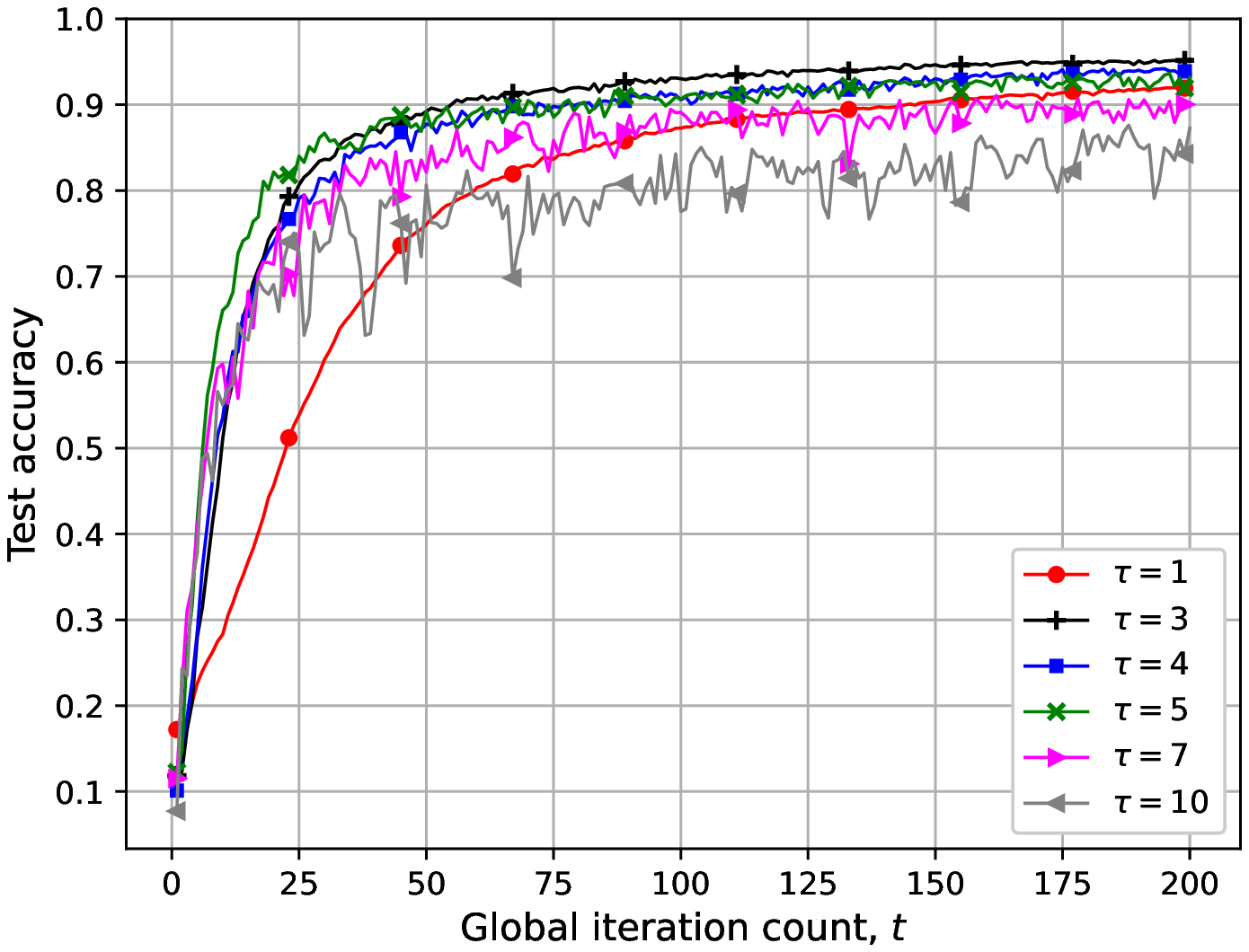}\vspace{0cm}
  \caption{$P^{\rm{dl}} = 10^2$}
  \label{Fig_Analog_high_power}
\end{subfigure}
\caption{Accuracy of analog downlink for the non-iid data distribution with $M=40$, $n^{\rm{dl}}=n^{\rm{ul}}=d/2$, $\sigma^{\rm{dl}} = \sigma^{\rm{ul}} = 1$, $P^{\rm{ul}} = 10$, $\lambda_{{\rm{thr}}} (t) = 10^{-4}$, $\forall t$, and $\left| \xi_m^i (t) \right| = 500$, $\forall i, m, t$.}
\label{Fig_Analog_low_high_power}
\end{figure}
The same observation can be made also for the performance of the digital downlink approach in the iid case and the analog downlink approach in the non-iid scenario. 
The reason for this opposite behavior is that, in contrast to the digital downlink approach, with the analog approach the devices have a relatively good estimate of $\boldsymbol{\theta}(t)$. 
For the analog downlink approach with sufficiently many devices, i.e., $M=40$, the best $\tau$ value for the iid case is larger than that for the non-iid case. 
This is intuitive since increasing $\tau$ excessively for the non-iid case provides biased local updates at the devices, which is due to the biased local datasets, with a relatively poor similarity. 
On the other hand, the digital downlink approach for $M=40$ shows the opposite behavior, which is due to the relatively inaccurate estimate of $\boldsymbol{\theta}(t)$ at the devices. 
In this case, for the iid scenario, in which the local data is homogeneous, the inaccuracy of the model parameter vector estimate harms the performance when a relatively large number of local SGD iterations are performed for both $M$ values.
Whereas, for $M=40$ in the non-iid scenario, a relatively small $\tau$ might not provide reliable local updates, since the local training dataset is biased and a relatively good estimate of $\boldsymbol{\theta}(t)$ is not available to rely on. 
On the other hand, for the digital approach with $M=20$, where devices receive a more accurate estimate of $\boldsymbol{\theta}(t)$, due to the higher achievable common rate, a relatively small $\tau$ value provides a better performance.
A similar observation is made for the analog downlink approach with $M=20$ devices in the iid case, where a relatively small $\tau$, $\tau=2$, provides the best performance. 
This is due to the fact that, having less devices for training, where each device performs local updates using homogeneous local data and a distinct noisy version of the global model, the chance of having the noise in the local updates cancelled out at the aggregation phase at the PS reduces when a relatively large $\tau$ is used for local updates. 
We provide a more in-depth investigation of the impact of number of local SGD iterations on the performance of the analog downlink approach in Figures \ref{Fig_Analog_low_high_power} and \ref{Fig_Converge_low_high_power}.
We remark here that the randomness in the experiments also have an impact on the experimental results presented here.

\begin{figure}[t!]
\centering
\begin{subfigure}{.5\textwidth}
  \centering
  \includegraphics[scale=0.55,trim={19pt 7pt 46pt 38pt},clip]{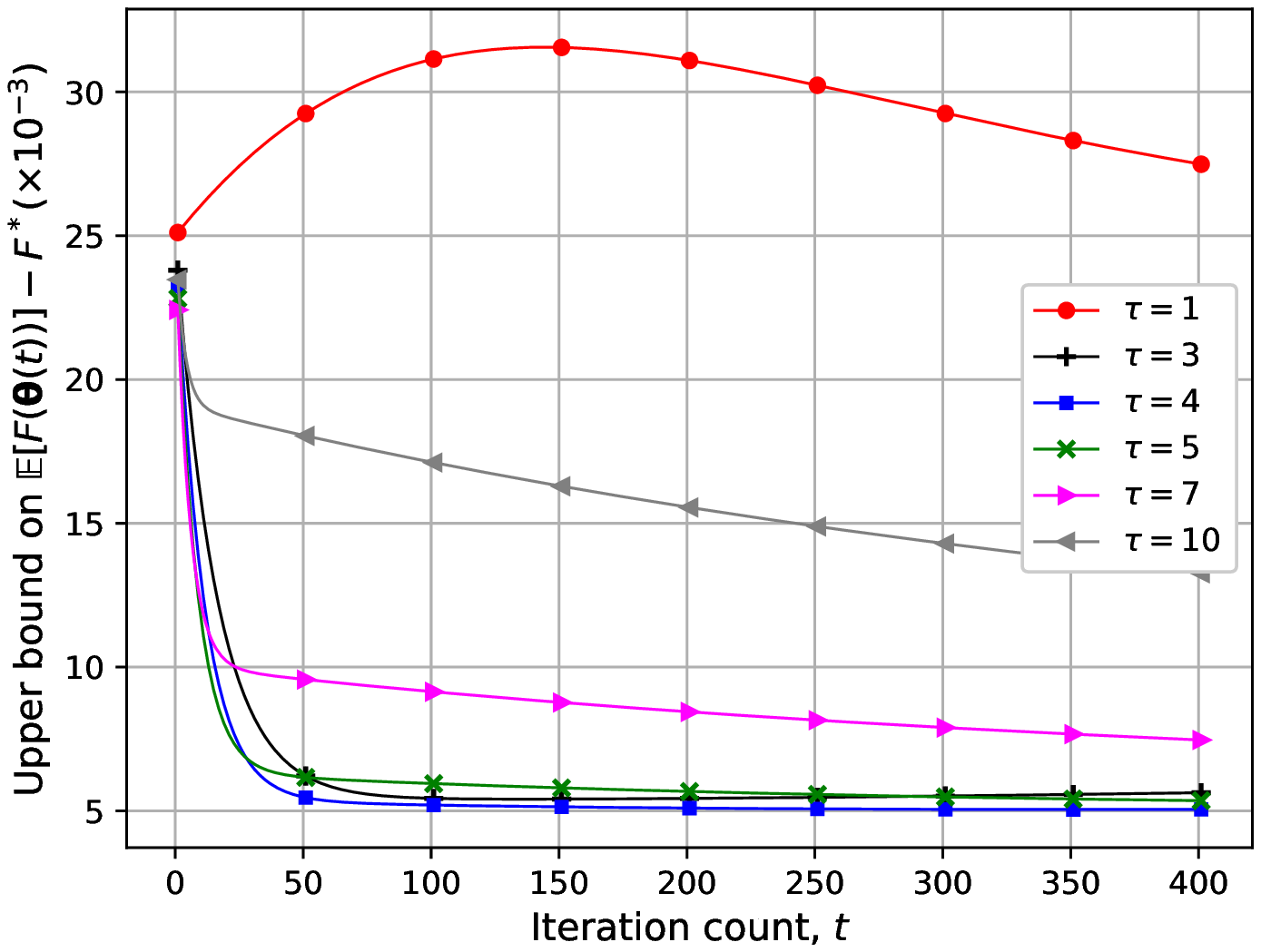}\vspace{0cm}
  \caption{$P^{\rm{dl}} = 10$}
  \label{Fig_Converge_low_power}
\end{subfigure}%
\begin{subfigure}{.5\textwidth}
  \centering
  \includegraphics[scale=0.55,trim={19pt 7pt 46pt 38pt},clip]{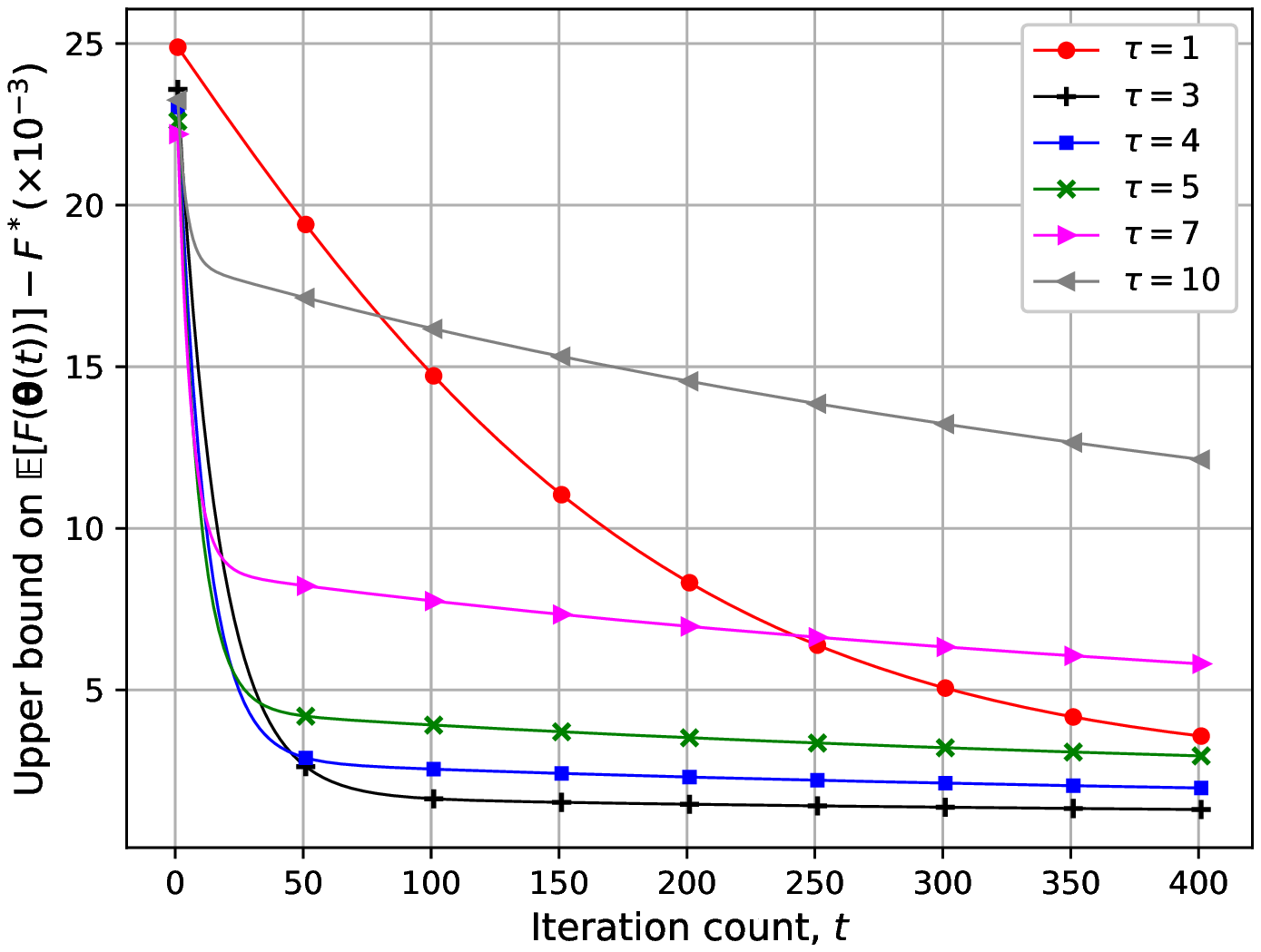}\vspace{0cm}
  \caption{$P^{\rm{dl}} = 10^2$}
  \label{Fig_Converge_high_power}
\end{subfigure}
\caption{Upper bound on $\mathbb{E} \left[ F( \boldsymbol{\theta} (t)) \right] - F^*$ for analog downlink for different $\tau$ values, $\tau \in \{ 1, 3, 4, 5, 7, 10\}$, considering non-iid data distribution with $(G^2, \Gamma) = (100, 50)$, for $\eta (t) = \frac{\min\left\{ \frac{\mu}{\mu+1}, \frac{1}{\mu \tau}\right\}}{(10^{-3}t + 1)}$, $\forall t$, $M=40$, $\mu = 0.2$, $L=10$, $\left\| {\boldsymbol{\theta}} (0) - {\boldsymbol{\theta}}^* \right\|_2^2 = 5 \times 10^3$, and $Z^2 = 2 \times 10^4$.}
\label{Fig_Converge_low_high_power}
\end{figure}

In Fig. \ref{Fig_Analog_low_high_power} we study the impact of $\tau$ on the performance of the analog downlink approach focusing on the non-iid data distribution for two different transmit power levels $P^{\rm{dl}} \in \{ 10, 10^2\}$ at the PS with $\tau \in \{1, 3, 4, 5, 7, 10\}$ and $M=40$ devices. 
We note that with a higher $P^{\rm{dl}}$ the devices receive a better/less noisy estimate of $\boldsymbol{\theta} (t)$. 
Observe that, for a smaller $P^{\rm{dl}}$, $P^{\rm{dl}}=10$, $\tau=4$ provides the best performance, while for $P^{\rm{dl}}=10^2$, the best performance is achieved for $\tau=3$. 
Therefore, for the non-iid scenario, when having a less accurate estimate of $\boldsymbol{\theta} (t)$ at the devices, a larger number of local SGD iterations should be performed compared to having a more accurate estimate of $\boldsymbol{\theta} (t)$ at the devices. 
As discussed for the performance of the digital downlink approach in Fig. \ref{Fig_IID_nonIID}, a relatively small $\tau$ value might not provide the most reliable local updates for the non-iid scenario when a good estimate of $\boldsymbol{\theta} (t)$ is not available at the devices. 
This observation is corroborated in Fig. \ref{Fig_Converge_low_high_power}, which demonstrates the analytical results on the convergence rate bound of the analog downlink approach for the non-iid scenario for different $\tau$ values, $\tau \in \{1, 3, 4, 5, 7, 10\}$, with two $P^{\rm{dl}}$ values, $P^{\rm{dl}}\in \{10, 10^2\}$. 
We observe in this figure that, for $P^{\rm{dl}} = 10$, $\tau = 4$ provides the best performance in terms of the convergence speed and the final level of the average loss. 
Whereas, for $P^{\rm{dl}} = 10^2$, $\tau=3$ provides the lowest average loss, although it has a negligibly smaller convergence speed compared to $\tau=4, 5, 7$.

In Fig. \ref{Fig_converge_IID_nonIID}, we consider the analytical convergence result of the analog downlink approach for the iid and non-iid scenarios for various $\tau$ values, $\tau \in \{1, 3, 4, 5, 7, 10\}$. 
We observe that, for the iid scenario, considering both the convergence rate and the final average loss, $\tau=5$ provides the best performance, although it has a slightly smaller convergence speed compared to $\tau =7, 10$. 
On the other hand, we observe that a smaller $\tau$ value, $\tau = 3$, has the best performance in the non-iid scenario. 
This result corroborates the observation made in Fig. \ref{Fig_IID_nonIID} for the analog downlink approach with $M=40$ devices, in which a larger $\tau$ value should be used for a less biased data distribution to obtain the best performance. 
A relatively large $\tau$ for non-iid data results in a more biased/skewed local updates with less consensus.

There results suggest that a schedule for $\tau$ that depends on the iteration $t$ might work well in a wide range of scenarios.
Specifically, start with a larger $\tau$ and decrease it as $t$ increases.

\begin{figure}[t!]
\centering
\begin{subfigure}{.5\textwidth}
  \centering
  \includegraphics[scale=0.55,trim={19pt 7pt 46pt 38pt},clip]{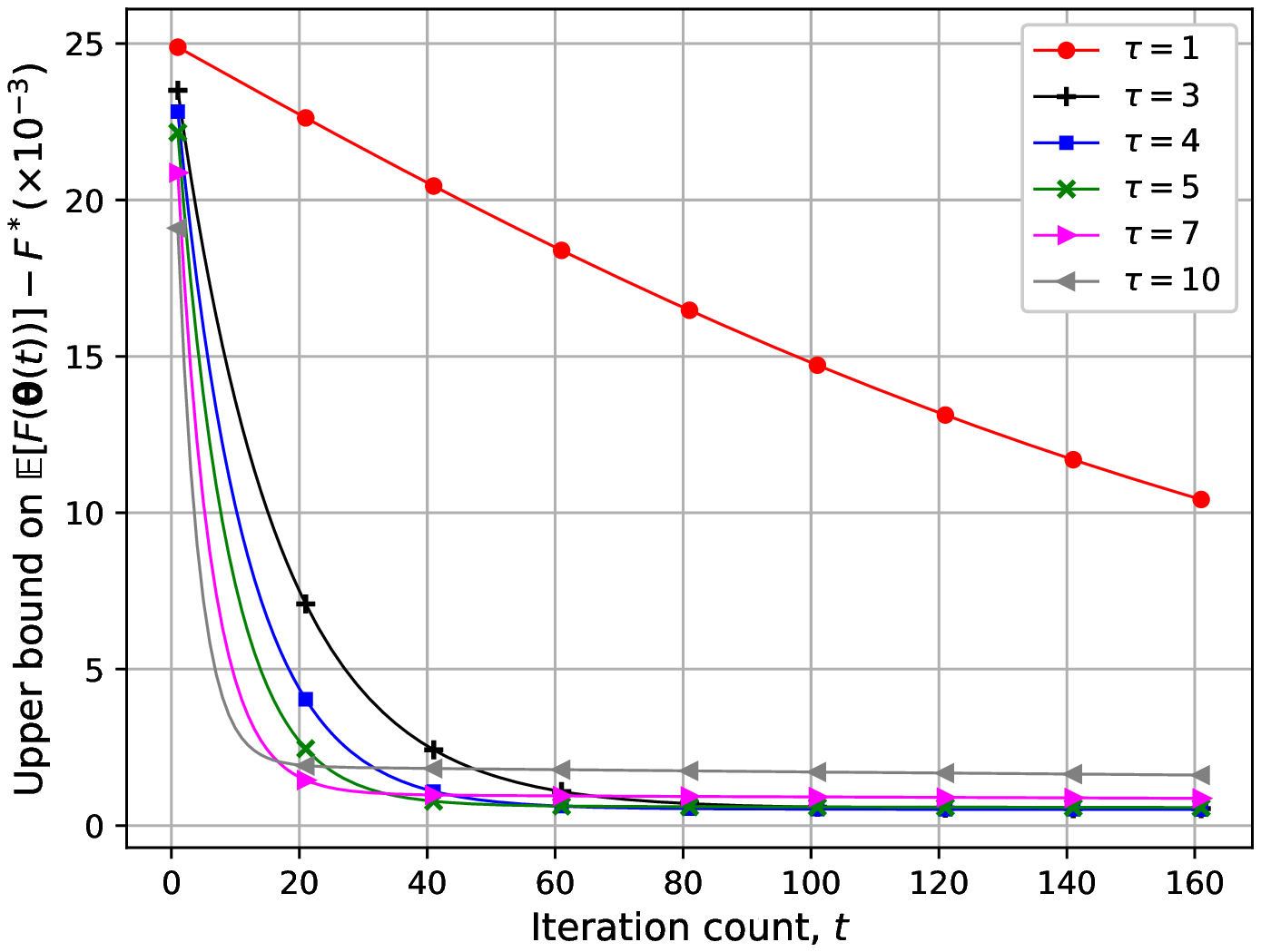}\vspace{0cm}
  \caption{Iid data distribution, $(G^2, \Gamma) = (10, 5)$}
  \label{Fig_converge_IID}
\end{subfigure}%
\begin{subfigure}{.5\textwidth}
  \centering
  \includegraphics[scale=0.55,trim={19pt 7pt 46pt 38pt},clip]{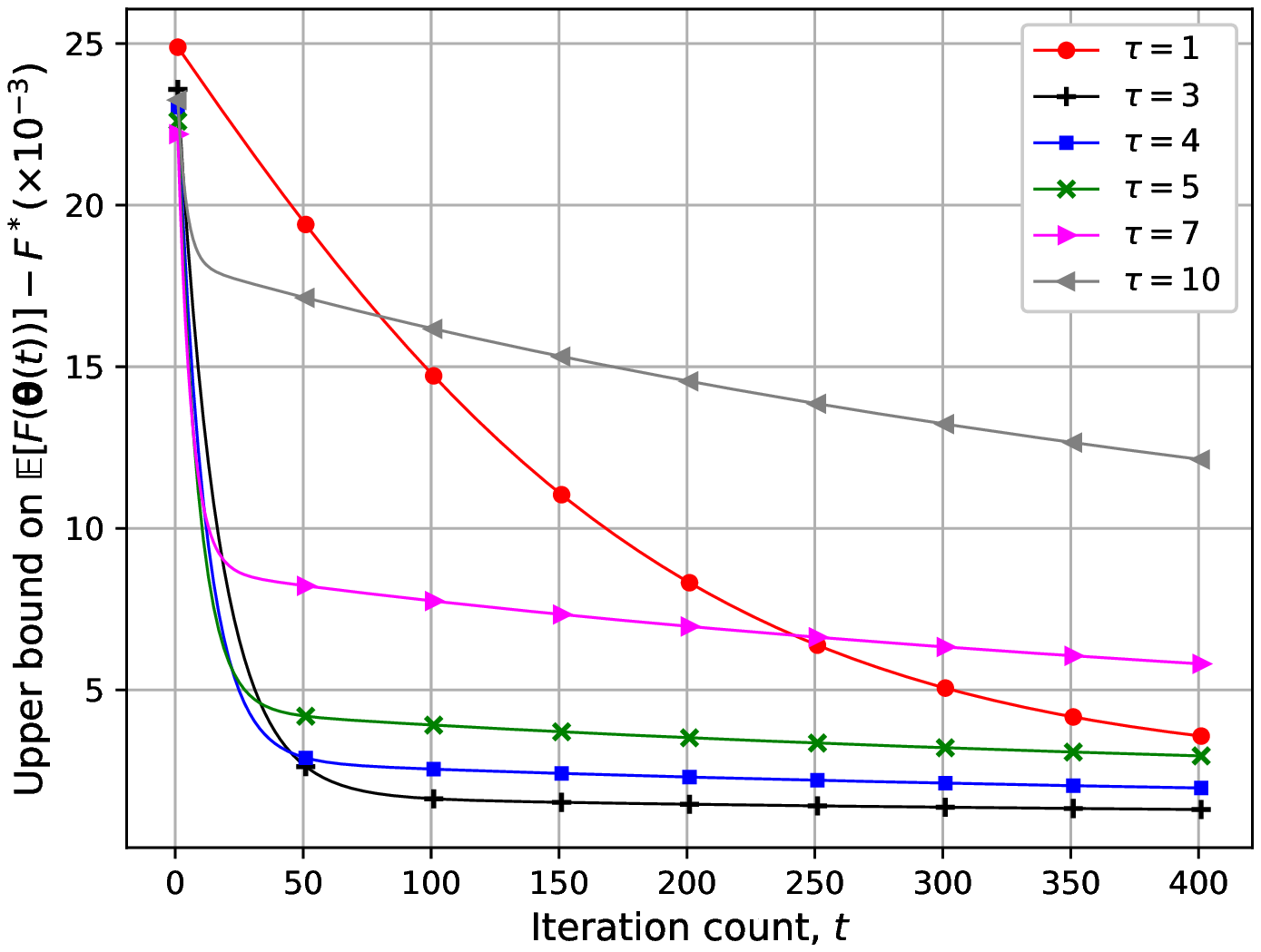}\vspace{0cm}
  \caption{Non-iid data distribution, $(G^2, \Gamma) = (100, 50)$}
  \label{Fig_converge_nonIID}
\end{subfigure}
\caption{Upper bound on $\mathbb{E} \left[ F( \boldsymbol{\theta} (t)) \right] - F^*$ for the analog downlink approach for different $\tau$ values, $\tau \in \{ 1, 3, 4, 5, 7, 10\}$, with $P^{\rm{dl}} = 10^2$, for $\eta (t) = \frac{\min\left\{ \frac{\mu}{\mu+1}, \frac{1}{\mu \tau}\right\}}{(10^{-3}t + 1)}$, $\forall t$, $M=40$, $\mu = 0.2$, $L=10$, $\left\| {\boldsymbol{\theta}} (0) - {\boldsymbol{\theta}}^* \right\|_2^2 = 5 \times 10^3$, and $Z^2 = 2 \times 10^4$.}
\label{Fig_converge_IID_nonIID}
\end{figure}

%\vspace{-2cm}
\section{Conclusions}\label{SecConc}
We have studied FEEL, where the PS with a limited power budget transmits the model parameter vector to the wireless devices over a bandwidth-limited fading broadcast channel. 
We have proposed digital and analog transmission approaches for the PS-to-devices transmission. 
With the digital approach, the PS quantizes the global model update, with respect to the global model estimate at the devices, with the knowledge of the highest common rate sustainable over the downlink broadcast channel. 
For the analysis, we have utilized a capacity achieving channel code to broadcast the same estimate of the global model update to all the devices. 
On the other hand, with the analog approach, the PS broadcasts the global model vector in an uncoded manner without employing any channel code, and the devices receive different estimates of the global model through independent wireless connections. 
%This approach provides a higher level of privacy since $\boldsymbol{\theta} (t)$ experiences different wireless connections with independent channel gains and additive noise to reach the devices.  
In both approaches, the devices perform multiple local SGD iterations with respect to their global model estimates utilizing their local datasets. 
The power-limited wireless devices then transmit their local model updates to the PS over a bandwidth-limited fading MAC in an analog fashion, whose superiority over digital transmission for the uplink has been shown in the literature \cite{MohammadDenizDSGDCS,FLTWCMohammadDenizFading,KaibinParallelWork}. 
We have also provided a convergence analysis for the analog downlink approach to study the impact of imperfect downlink transmission, leading to noisy estimates of the global model at the devices, on the performance of FL, where for the ease of analysis we have assumed that the uplink transmission is error-free. 
Numerical experiments on the MNIST dataset have shown a significant improvement of the analog downlink approach over its digital counterpart, where the improvement is more pronounced for the non-iid data scenario. 
The analog downlink approach benefits from providing the devices with different estimates of the global model with the quality of these estimates depending on their downlink channel conditions, in which case the destructive effect of the devices with relatively worse channel conditions, and consequently less accurate estimates, can be alleviated by the devices with better channel conditions. 
However, with the digital downlink approach, the devices receive the same estimate of the model parameter vector with a common rate limited by the capacity of the worst device. 
Therefore, it is likely that all the devices perform local SGD iterations using an inaccurate estimate of the global model.
%Another inherent advantage of the analog approach is the enhanced privacy due to the random noise added to the model parameter vector over downlink at each device.
Both the experimental and analytical results have shown that a smaller number of local SGD iterations should be performed to obtain the best performance of the analog downlink approach for non-iid data compared to iid data. 
Also, for non-iid data, by increasing the transmit power at the PS, which leads to a more accurate global model estimate at the devices, a smaller number of local SGD iterations should be performed at the devices.

\appendices

\section{Proof of Theorem \ref{A_Theoremtheta_thetastarKequalM}}\label{A_AppTheorem}

The global model parameter vector for the analog downlink approach is updated as
\begin{align}\label{A_AppFDPModUpdateTheta}
\boldsymbol{\theta} (t+1) = {\boldsymbol{\theta}} (t) +  \sum\nolimits_{m =1}^M \frac{B_m}{B} \Delta \boldsymbol{\theta}_m (t).
\end{align}
%we define the following auxiliary variable:
%\begin{align}\label{A_AppFDPThetaTilde}
%\boldsymbol{v} (t+1) = {\boldsymbol{\theta}} (t) + \frac{1}{M} \sum\limits_{m =1}^{M} \Delta \boldsymbol{\theta}_m (t).
%\end{align}
%For ease of presentation, we drop the expectation notation $\mathbb{E} [\cdot]$ for the analysis; however, all the following terms are averaged over the stochastic gradients and the quantization technique. 
We have 
\begin{align}\label{A_AppFDP_1}
\mathbb{E} \left[ \left\| \boldsymbol{\theta} (t+1) - {\boldsymbol{\theta}}^* \right\|_2^2 \right] = & \mathbb{E} \left[ \left\| {\boldsymbol{\theta}} (t) - {\boldsymbol{\theta}}^* \right\|_2^2 \right] + \mathbb{E} \bigg[ \left\| \sum\nolimits_{m =1}^M \frac{B_m}{B} \Delta \boldsymbol{\theta}_m (t) \right\|_2^2 \bigg] \nonumber\\
& \qquad \qquad \qquad \qquad \quad + 2 \mathbb{E} \left[ \langle {\boldsymbol{\theta}} (t) - {\boldsymbol{\theta}}^* , \sum\nolimits_{m =1}^M \frac{B_m}{B} \Delta \boldsymbol{\theta}_m (t) \rangle \right].  
\end{align}
Next we bound the last two terms on the right hand side (RHS) of \eqref{A_AppFDP_1}.

From the convexity of $\left\| \cdot \right\|_2^2$, it follows that
\newcommand\firstinequal{\mathrel{\overset{\makebox[0pt]{\mbox{\normalfont\tiny\sffamily (a)}}}{\le}}}
\begin{align}\label{AppLemmaTemr_2_Eq_1_2}
&\mathbb{E} \bigg[ \left\| \sum\nolimits_{m =1}^M \frac{B_m}{B} \Delta \boldsymbol{\theta}_m (t) \right\|_2^2 \bigg] \le  \sum\nolimits_{m =1}^{M} \frac{B_m}{B} \mathbb{E} \left[ \left\| \Delta \boldsymbol{\theta}_m (t) \right\|_2^2 \right] \qquad \qquad \qquad \qquad \qquad \nonumber
\end{align}
\begin{align}
& \qquad \qquad \quad = \eta^2(t) \sum\nolimits_{m =1}^{M} \frac{B_m}{B} \mathbb{E} \left[ \left\| \sum\nolimits_{i=1}^{\tau} \nabla F_m \left( \boldsymbol{\theta}_m^i (t), \xi_m^i (t) \right) \right\|_2^2 \right] \nonumber\\
& \qquad \qquad \quad  \le \eta^2(t) \tau \sum\nolimits_{m =1}^{M} \sum\nolimits_{i=1}^{\tau} \frac{B_m}{B} \mathbb{E} \left[ \left\| \nabla F_m \left( \boldsymbol{\theta}_m^i (t), \xi_m^i (t) \right) \right\|_2^2 \right] \firstinequal \eta^2(t) \tau^2 G^2,
\end{align}
where (a) follows from Assumption \ref{AssumpBoundedVarGradient}. 

We rewrite the third term on the RHS of \eqref{A_AppFDP_1} as follows:
\begin{align}\label{A_AppLemmaTemr_2_Eq_2}
&2 \mathbb{E} \left[ \langle {\boldsymbol{\theta}} (t) - {\boldsymbol{\theta}}^* ,  \sum\nolimits_{m =1}^{M} \frac{B_m}{B} \Delta \boldsymbol{\theta}_m (t) \rangle \right] \nonumber\\
& \qquad = 2 \eta(t) \sum\nolimits_{m=1}^{M} \frac{B_m}{B} \mathbb{E} \left[ \langle {\boldsymbol{\theta}}^* - {\boldsymbol{\theta}} (t) , \sum\nolimits_{i=1}^{\tau} \nabla F_m \left( \boldsymbol{\theta}_m^i (t), \xi_m^i (t) \right) \rangle \right] \nonumber\\
& \qquad = 2 \eta(t) \sum\nolimits_{m=1}^{M} \frac{B_m}{B} \mathbb{E} \left[ \langle {\boldsymbol{\theta}}^* - {\boldsymbol{\theta}} (t) , \nabla F_m \left( {\boldsymbol{\theta}} (t) + \tilde{\boldsymbol{z}}_m^{\rm{dl}} (t), \xi_m^1 (t) \right) \rangle \right] \nonumber\\
& \qquad \quad + 2 \eta(t) \sum\nolimits_{m=1}^{M} \frac{B_m}{B} \mathbb{E} \left[ \langle {\boldsymbol{\theta}}^* - {\boldsymbol{\theta}} (t) , \sum\nolimits_{i=2}^{\tau} \nabla F_m \left( \boldsymbol{\theta}_m^i (t), \xi_m^i (t) \right) \rangle \right]. 
\end{align}
We have
\begin{align}\label{A_AppLemmaTemr_2_Eq_3}
& 2 \eta(t) \sum\nolimits_{m=1}^{M} \frac{B_m}{B} \mathbb{E} \left[ \langle {\boldsymbol{\theta}}^* - {\boldsymbol{\theta}} (t) , \nabla F_m \left( {\boldsymbol{\theta}} (t) + \tilde{\boldsymbol{z}}_m^{\rm{dl}} (t), \xi_m^1 (t) \right) \rangle \right] \nonumber\\
& = 2 \eta(t) \sum\nolimits_{m=1}^{M} \frac{B_m}{B} \mathbb{E} \left[ \langle {\boldsymbol{\theta}}^* - {\boldsymbol{\theta}} (t) , \nabla F_m \left( {\boldsymbol{\theta}} (t), \xi_m^1 (t) \right) \rangle \right] \nonumber\\
& \;\;\; + 2 \eta(t) \sum\nolimits_{m=1}^{M} \frac{B_m}{B} \mathbb{E} \left[ \langle {\boldsymbol{\theta}}^* - {\boldsymbol{\theta}} (t) , \nabla F_m \left( {\boldsymbol{\theta}} (t) + \widetilde{\boldsymbol{z}}_m^{\rm{dl}} (t), \xi_m^1 (t) \right) - \nabla F_m \left( {\boldsymbol{\theta}} (t), \xi_m^1 (t) \right) \rangle \right].
\end{align}
In the following, we bound the two terms on the RHS of \eqref{A_AppLemmaTemr_2_Eq_3}. 
We have
\newcommand\thirdequal{\mathrel{\overset{\makebox[0pt]{\mbox{\normalfont\tiny\sffamily (a)}}}{=}}}
\newcommand\thirdinequal{\mathrel{\overset{\makebox[0pt]{\mbox{\normalfont\tiny\sffamily (b)}}}{\le}}}
\begin{align}\label{AppLemmaTemr_2_Eq_3}
& 2 \eta(t) \sum\nolimits_{m=1}^{M} \frac{B_m}{B} \mathbb{E} \left[ \langle {\boldsymbol{\theta}}^* - {\boldsymbol{\theta}} (t) , \nabla F_m \left( {\boldsymbol{\theta}} (t), \xi_m^1 (t) \right) \rangle \right] \nonumber\\
& \qquad \qquad \qquad \qquad \thirdequal 2 \eta(t) \sum\nolimits_{m=1}^{M} \frac{B_m}{B} \mathbb{E} \left[ \langle {\boldsymbol{\theta}}^* - {\boldsymbol{\theta}} (t) , \nabla F_m \left( {\boldsymbol{\theta}} (t) \right) \rangle \right] \nonumber\\
& \qquad \qquad \qquad \qquad \thirdinequal 2 \eta(t) \sum\nolimits_{m=1}^{M} \frac{B_m}{B} \mathbb{E} \left[ F_m (\boldsymbol{\theta}^*) - F_m \left( {\boldsymbol{\theta}} (t) \right) - \frac{\mu}{2} \left\| {\boldsymbol{\theta}} (t) - {\boldsymbol{\theta}}^* \right\|_2^2 \right] \nonumber\\
& \qquad \qquad \qquad \qquad = 2 \eta(t) \left( F^* - \mathbb{E} \left[ F\left( {\boldsymbol{\theta}} (t) \right) \right] - \frac{\mu}{2} \mathbb{E} \Big[ \left\| {\boldsymbol{\theta}} (t) - {\boldsymbol{\theta}}^* \right\|_2^2 \Big] \right),
\end{align}
where (a) and (b) follow from \eqref{AverageStochGradientEst} and Assumption \ref{AssumpStrongConvexLoss}, respectively. 
Also, from Cauchy-Schwarz inequality, we have
\begin{align}\label{A_AppLemmaTemr_2_Eq_5}
& 2 \eta(t) \sum\nolimits_{m=1}^{M} \frac{B_m}{B} \mathbb{E} \left[ \langle {\boldsymbol{\theta}}^* - {\boldsymbol{\theta}} (t) , \nabla F_m \left( {\boldsymbol{\theta}} (t) + \widetilde{\boldsymbol{z}}_m^{\rm{dl}} (t), \xi_m^1 (t) \right) - \nabla F_m \left( {\boldsymbol{\theta}} (t), \xi_m^1 (t) \right) \rangle \right]\nonumber\\
& \le \eta^2(t) \mathbb{E} \left[ \left\| {\boldsymbol{\theta}} (t) - {\boldsymbol{\theta}}^* \right\|_2^2 \right] + \mathbb{E} \left[ \left\| \sum\nolimits_{m=1}^{M} \frac{B_m}{B} \left( \nabla F_m ( \boldsymbol{\theta} (t) + \widetilde{\boldsymbol{z}}_m^{\rm{dl}} (t), \xi_m^1 (t) ) - \nabla F_m ( \boldsymbol{\theta} (t), \xi_m^1 (t) ) \right) \right\|^2 \right]\nonumber\\
& \firstinequal \eta^2(t) \mathbb{E} \left[ \left\| {\boldsymbol{\theta}} (t) - {\boldsymbol{\theta}}^* \right\|_2^2 \right] + \frac{ Z^2}{M \sigma^{\rm{dl}} {P}^{\rm{dl}}},
\end{align}
where (a) follows from Assumption \ref{AssumpBoundThetaHatmAnalog}. 
Substituting \eqref{AppLemmaTemr_2_Eq_3} and \eqref{A_AppLemmaTemr_2_Eq_5} into \eqref{A_AppLemmaTemr_2_Eq_3} yields
\begin{align}\label{A_AppLemmaTemr_2_Eq_6}
& 2 \eta(t) \sum\nolimits_{m=1}^{M} \frac{B_m}{B} \mathbb{E} \left[ \langle {\boldsymbol{\theta}}^* - {\boldsymbol{\theta}} (t) , \nabla F_m \left( {\boldsymbol{\theta}} (t) + \widetilde{\boldsymbol{z}}_m^{\rm{dl}} (t), \xi_m^1 (t) \right) \rangle \right] \nonumber\\
& \quad \quad \; \le - \mu \eta(t) \left( 1 - \eta(t)/\mu \right) \mathbb{E} \left[ \left\| {\boldsymbol{\theta}} (t) - {\boldsymbol{\theta}}^* \right\|_2^2 \right] + \frac{ Z^2}{M \sigma^{\rm{dl}} {P}^{\rm{dl}}} + 2 \eta(t) \left( F^* - \mathbb{E} \left[ F\left( {\boldsymbol{\theta}} (t) \right) \right] \right). 
\end{align}

\begin{lemma}\label{LemmaTermE}
For $0 < \eta(t) \le \frac{\mu}{\mu + 1}$, we have
\begin{align}\label{A_AppLemmaTemr_2_Eq_7}
& 2 \eta(t) \sum\nolimits_{m=1}^{M} \frac{B_m}{B} \mathbb{E} \left[ \langle {\boldsymbol{\theta}}^* - {\boldsymbol{\theta}} (t) , \sum\nolimits_{i=2}^{\tau} \nabla F_m \left( \boldsymbol{\theta}_m^i (t), \xi_m^i (t) \right) \rangle \right] \nonumber\\
& \qquad   \le - \mu \eta(t) (1 - \eta(t)) (\tau - 1) \mathbb{E} \left[ \left\| {\boldsymbol{\theta}} (t) - \boldsymbol{\theta}^* \right\|_2^2 \right] \nonumber\\  
& \qquad  \quad + (1+ \mu (1 - \eta(t))) \eta^2(t) G^2 \frac{\tau (\tau-1)(2\tau-1)}{6} + 2 \eta(t) (\tau - 1) \Gamma  \nonumber\\
& \qquad  \quad + \left( \eta^2 (t) + 1 \right) \left( \tau - 1 \right) G^2 + 2 \eta (t) \sum\nolimits_{m=1}^{M} \sum\nolimits_{i=2}^{\tau} \frac{B_m}{B} \left( F_m^* - \mathbb{E} \left[ F_m({\boldsymbol{\theta}}_m^i (t)) \right] \right). 
\end{align}
\end{lemma}
\begin{proof}
See Appendix \ref{AppProofLemmaTermE}. 
\end{proof}

By substituting \eqref{A_AppLemmaTemr_2_Eq_6} and \eqref{A_AppLemmaTemr_2_Eq_7} in \eqref{A_AppLemmaTemr_2_Eq_2}, it follows that
\begin{align}\label{A_AppLemmaTemr_2_Eq_8}
&2 \mathbb{E} \left[ \langle {\boldsymbol{\theta}} (t) - {\boldsymbol{\theta}}^* ,  \sum\nolimits_{m =1}^{M} \frac{B_m}{B} \Delta \boldsymbol{\theta}_m (t) \rangle \right] \le - \mu \eta (t) \left( \tau - \eta(t) (\tau - 1 + 1/\mu) \right) \mathbb{E} \left[ \left\| {\boldsymbol{\theta}} (t) - \boldsymbol{\theta}^* \right\|_2^2 \right]\nonumber\\
& + \frac{ Z^2}{M \sigma^{\rm{dl}} {P}^{\rm{dl}}} + (1+ \mu (1 - \eta(t))) \eta^2(t) G^2 \frac{\tau (\tau-1)(2\tau-1)}{6} + \left( \eta^2 (t) + 1 \right) \left( \tau - 1 \right) G^2  \nonumber\\
& + 2 \eta(t) (\tau - 1) \Gamma + 2 \eta (t) \sum\nolimits_{m=1}^{M} \sum\nolimits_{i=2}^{\tau} \frac{B_m}{B} \left( F_m^* - \mathbb{E} \left[ F_m({\boldsymbol{\theta}}_m^i (t)) \right] \right) + 2 \eta(t) \left( F^* - \mathbb{E} \left[ F\left( {\boldsymbol{\theta}} (t) \right) \right] \right),
\end{align}
which together with the inequality in \eqref{AppLemmaTemr_2_Eq_1_2}, according to \eqref{A_AppFDP_1}, the following upper bound on $\mathbb{E} \left[ \left\| \boldsymbol{\theta} (t+1) - {\boldsymbol{\theta}}^* \right\|_2^2 \right]$ is obtained:
\begin{align}\label{A_AppLemmaTemr_2_Eq_9}
& \mathbb{E} \left[ \left\| \boldsymbol{\theta} (t+1) - {\boldsymbol{\theta}}^* \right\|_2^2 \right] \le \left( 1 - \mu \eta (t) \left( \tau - \eta(t) (\tau - 1 + 1/\mu) \right) \right) \mathbb{E} \left[ \left\| {\boldsymbol{\theta}} (t) - \boldsymbol{\theta}^* \right\|_2^2 \right] + \frac{ Z^2}{M \sigma^{\rm{dl}} {P}^{\rm{dl}}} \nonumber\\
& \; + (1+ \mu (1 - \eta(t))) \eta^2(t) G^2 \frac{\tau (\tau-1)(2\tau-1)}{6} + \left( \tau - 1 + \eta^2 (t) \left( \tau^2 + \tau - 1 \right) \right) G^2  \nonumber\\
& \; + 2 \eta(t) (\tau - 1) \Gamma + 2 \eta (t) \sum\nolimits_{m=1}^{M} \sum\nolimits_{i=2}^{\tau} \frac{B_m}{B} \left( F_m^* - \mathbb{E} \left[ F_m({\boldsymbol{\theta}}_m^i (t)) \right] \right) + 2 \eta(t) \left( F^* - \mathbb{E} \left[ F\left( {\boldsymbol{\theta}} (t) \right) \right] \right)\nonumber\\
& \firstinequal \left( 1 - \mu \eta (t) \left( \tau - \eta(t) (\tau - 1 + 1/\mu) \right) \right) \mathbb{E} \left[ \left\| {\boldsymbol{\theta}} (t) - \boldsymbol{\theta}^* \right\|_2^2 \right] + \frac{ Z^2}{M \sigma^{\rm{dl}} {P}^{\rm{dl}}} \nonumber\\
& \; + (1+ \mu (1 - \eta(t))) \eta^2(t) G^2 \frac{\tau (\tau-1)(2\tau-1)}{6} \nonumber\\
& \; + \left( \tau - 1 + \eta^2 (t) \left( \tau^2 + \tau - 1 \right) \right) G^2 + 2 \eta(t) (\tau - 1) \Gamma,
\end{align}
where (a) follows sine $F^* - F(\boldsymbol{\theta} (t)) \le 0$, $\forall t$, and $F_m^* - F_m({\boldsymbol{\theta}}_m^i (t)) \le 0$, $\forall m, i, t$.
It is trivial to prove Theorem \ref{A_Theoremtheta_thetastarKequalM} from the inequality in \eqref{A_AppLemmaTemr_2_Eq_9} for $0 < \eta(t) \le \min \left\{ \frac{\mu}{\mu + 1}, \frac{1}{\mu \tau} \right\}$, $\forall t$.

\section{Proof of Lemma \ref{LemmaTermE}}\label{AppProofLemmaTermE}
We have 
\begin{align}\label{EQ_LemmaTermE_app_1}
& 2 \eta(t) \sum\nolimits_{m=1}^{M} \sum\nolimits_{i=2}^{\tau} \frac{B_m}{B} \mathbb{E} \left[ \langle \boldsymbol{\theta}^* - {\boldsymbol{\theta}} (t) , \nabla F_m \left( \boldsymbol{\theta}_m^i (t), \xi_m^i (t) \right) \rangle \right] \nonumber\\
& = 2 \eta(t) \sum\nolimits_{m=1}^{M} \sum\nolimits_{i=2}^{\tau} \frac{B_m}{B} \mathbb{E} \left[ \langle \boldsymbol{\theta}_m^i (t) - {\boldsymbol{\theta}} (t) , \nabla F_m \left( \boldsymbol{\theta}_m^i (t), \xi_m^i (t) \right) \rangle \right]\nonumber\\
& \;\; \;\;+ 2 \eta(t) \sum\nolimits_{m=1}^{M} \sum\nolimits_{i=2}^{\tau} \frac{B_m}{B} \mathbb{E} \left[ \langle \boldsymbol{\theta}^* - \boldsymbol{\theta}_m^i (t),  \nabla F_m \left( \boldsymbol{\theta}_m^i (t), \xi_m^i (t) \right) \rangle \right].  
\end{align}
For the first term on the RHS of \eqref{EQ_LemmaTermE_app_1}, we have
\begin{align}\label{AppLemmaTemr_1_D_Eq_1}
& 2 \eta(t) \sum\nolimits_{m=1}^{M} \sum\nolimits_{i=2}^{\tau} \frac{B_m}{B} \mathbb{E} \left[ \langle \boldsymbol{\theta}_m^i (t) - {\boldsymbol{\theta}} (t) , \nabla F_m \left( \boldsymbol{\theta}_m^i (t), \xi_m^i (t) \right) \rangle \right] \nonumber\\
& \qquad \qquad = 2 \eta(t) \sum\nolimits_{m=1}^{M} \sum\nolimits_{i=2}^{\tau} \frac{B_m}{B} \mathbb{E} \left[ \langle \boldsymbol{\theta}_m^i (t) - {\boldsymbol{\theta}}_m^1 (t) , \nabla F_m \left( \boldsymbol{\theta}_m^i (t), \xi_m^i (t) \right) \rangle \right] \nonumber\\
& \qquad \qquad \quad + 2 \eta(t) \sum\nolimits_{m=1}^{M} \sum\nolimits_{i=2}^{\tau} \frac{B_m}{B} \mathbb{E} \left[ \langle \widetilde{\boldsymbol{z}}^{\rm{dl}}_m (t) , \nabla F_m \left( \boldsymbol{\theta}_m^i (t), \xi_m^i (t) \right) \rangle \right].
\end{align}
From Cauchy-Schwarz inequality, we have
\begin{align}\label{AppLemmaTemr_1_D_Eq_11}
&2 \eta(t) \sum\nolimits_{m=1}^{M} \sum\nolimits_{i=2}^{\tau} \frac{B_m}{B} \mathbb{E} \left[ \langle \boldsymbol{\theta}_m^i (t) - {\boldsymbol{\theta}}_m^1 (t) , \nabla F_m \left( \boldsymbol{\theta}_m^i (t), \xi_m^i (t) \right) \rangle \right]\nonumber\\
& \; \; \; \quad \le \eta(t) \sum\nolimits_{m=1}^{M} \sum\nolimits_{i=2}^{\tau} \frac{B_m}{B} \mathbb{E} \bigg[ \frac{1}{\eta(t)} \left\| \boldsymbol{\theta}_m^i (t) - {\boldsymbol{\theta}}_m^1 (t) \right\|_2^2 + \eta(t) \left\| \nabla F_m \left( \boldsymbol{\theta}_m^i (t), \xi_m^i (t) \right) \right\|_2^2 \bigg]\nonumber\\
& \; \; \; \quad \firstinequal \sum\nolimits_{m=1}^{M} \sum\nolimits_{i=2}^{\tau} \frac{B_m}{B} \mathbb{E} \left[ \left\| \boldsymbol{\theta}_m^i (t) - {\boldsymbol{\theta}}_m^1 (t) \right\|_2^2 \right] + \eta^2 (t) \left( \tau - 1 \right) G^2,
\end{align}
and
\begin{align}\label{AppLemmaTemr_1_D_Eq_12}
& 2 \eta(t) \sum\nolimits_{m=1}^{M} \sum\nolimits_{i=2}^{\tau} \frac{B_m}{B} \mathbb{E} \left[ \langle \widetilde{\boldsymbol{z}}^{\rm{dl}}_m (t) , \nabla F_m \left( \boldsymbol{\theta}_m^i (t), \xi_m^i (t) \right) \rangle \right] \nonumber\\
& \; \; \; \quad \le \eta(t) \sum\nolimits_{m=1}^{M} \sum\nolimits_{i=2}^{\tau} \frac{B_m}{B} \mathbb{E} \bigg[ \eta(t) \left\| \widetilde{\boldsymbol{z}}^{\rm{dl}}_m (t) \right\|_2^2 + \frac{1}{\eta(t)} \left\| \nabla F_m \left( \boldsymbol{\theta}_m^i (t), \xi_m^i (t) \right) \right\|_2^2 \bigg] \nonumber\\
& \; \; \; \quad \firstinequal \eta^2(t) (\tau - 1) \sum\nolimits_{m=1}^{M} \frac{B_m}{B} \mathbb{E} \left[ \left\| \widetilde{\boldsymbol{z}}^{\rm{dl}}_m (t) \right\|_2^2 \right] + \left( \tau - 1 \right) G^2,
\end{align}
where (a) follows from Assumption \ref{AssumpBoundedVarGradient}.
Thus, the term on the left hand side (LHS) of \eqref{AppLemmaTemr_1_D_Eq_1} is bounded as
\begin{align}\label{AppLemmaTemr_1_D_Eq_1_bound}
& 2 \eta(t) \sum\nolimits_{m=1}^{M} \sum\nolimits_{i=2}^{\tau} \frac{B_m}{B} \mathbb{E} \left[ \langle \boldsymbol{\theta}_m^i (t) - {\boldsymbol{\theta}} (t) , \nabla F_m \left( \boldsymbol{\theta}_m^i (t), \xi_m^i (t) \right) \rangle \right]\nonumber \\
& \; \; \; \quad \le \sum\nolimits_{m=1}^{M} \sum\nolimits_{i=2}^{\tau} \frac{B_m}{B} \mathbb{E} \left[ \left\| \boldsymbol{\theta}_m^i (t) - {\boldsymbol{\theta}}_m^1 (t) \right\|_2^2 \right]  + \eta^2(t) (\tau - 1) \sum\nolimits_{m=1}^{M} \frac{B_m}{B} \mathbb{E} \left[ \left\| \widetilde{\boldsymbol{z}}^{\rm{dl}}_m (t) \right\|_2^2 \right] \nonumber\\
& \; \; \; \quad \quad + \left( \eta^2 (t) +1 \right) \left( \tau - 1 \right) G^2.
\end{align}
%We have
%\begin{align}\label{AppLemmaTemr_2_Eq_7}
%& \sum\limits_{m=1}^{M} \sum\limits_{i=2}^{\tau} \frac{B_m}{B} \mathbb{E} \left[ \left\| \boldsymbol{\theta}_m^i (t) - {\boldsymbol{\theta}} (t) \right\|_2^2 \right] \nonumber\\
%& \;\;\;= \eta^2(t) \sum\limits_{m=1}^{M} \sum\limits_{i=2}^{\tau} \frac{B_m}{B} \mathbb{E} \left[ \left\| \sum\nolimits_{j=1}^{i} \nabla F_m \left( \boldsymbol{\theta}_m^j (t), \xi_m^j (t) \right) \right\|_2^2 \right] \firstinequal \eta^2(t) G^2 \frac{\tau (\tau-1)(2\tau-1)}{6}, 
%\end{align}
%where (a) follows from the convexity of $\left\| \cdot \right\|_2^2$ and Assumption \ref{AssumpBoundedVarGradient}. 
%Plugging \eqref{AppLemmaTemr_2_Eq_7} into \eqref{AppLemmaTemr_1_D_Eq_1} yields  
%\begin{align}\label{AppLemmaTemr_1_D_Eq_2}
%& 2 \eta(t) \sum\nolimits_{m=1}^{M} \sum\nolimits_{i=2}^{\tau} \frac{B_m}{B} \mathbb{E} \left[ \langle \boldsymbol{\theta}_m^i (t) - {\boldsymbol{\theta}} (t) , \nabla F_m \left( \boldsymbol{\theta}_m^i (t), \xi_m^i (t) \right) \rangle \right]\nonumber \\
%& \; \; \; \qquad \qquad \qquad \qquad \qquad \qquad \le \eta^2(t) G^2 \frac{\tau (\tau-1)(2\tau-1)}{6} + \eta^2 (t) \left( \tau - 1 \right) G^2.
%\end{align}
From convexity of $\left\| \cdot \right\|_2^2$, we have
\begin{align}\label{Theta_m_iTheta_m_1}
& \sum\limits_{m=1}^{M} \sum\limits_{i=2}^{\tau} \frac{B_m}{B} \mathbb{E} \left[ \left\| \boldsymbol{\theta}_m^i (t) - {\boldsymbol{\theta}}_m^1 (t) \right\|_2^2 \right] = \eta^2(t) \sum\limits_{m=1}^{M} \sum\limits_{i=2}^{\tau} \frac{B_m}{B} \mathbb{E} \left[ \left\| \sum\nolimits_{j=1}^{i-1} \nabla F_m \left( \boldsymbol{\theta}_m^i (t), \xi_m^i (t) \right) \right\|_2^2 \right] \nonumber\\
& \le \eta^2(t) \sum\limits_{m=1}^{M} \sum\limits_{i=2}^{\tau} \frac{B_m}{B} (i-1) \sum\limits_{j=1}^{i-1} \mathbb{E} \left[ \left\|  \nabla F_m \left( \boldsymbol{\theta}_m^i (t), \xi_m^i (t) \right) \right\|_2^2 \right] \firstinequal \eta^2(t) G^2 \frac{\tau (\tau-1)(2\tau-1)}{6},
\end{align}
where (a) follows from Assumption \ref{AssumpBoundedVarGradient}.
For the second term on the RHS of \eqref{EQ_LemmaTermE_app_1}, we have
\begin{align}\label{EQ_LemmaTermE_D_app_1}
& 2 \eta(t) \sum\nolimits_{m=1}^{M} \sum\nolimits_{i=2}^{\tau} \frac{B_m}{B} \mathbb{E} \left[ \langle \boldsymbol{\theta}^* - \boldsymbol{\theta}_m^i (t) , \nabla F_m \left( \boldsymbol{\theta}_m^i (t), \xi_m^i (t) \right) \rangle \right] \nonumber\\
& \thirdequal 2 \eta(t) \sum\nolimits_{m=1}^{M} \sum\nolimits_{i=2}^{\tau} \frac{B_m}{B} \mathbb{E} \left[ \langle \boldsymbol{\theta}^* - \boldsymbol{\theta}_m^i (t) , \nabla F_m \left( \boldsymbol{\theta}_m^i (t) \right) \rangle \right]\nonumber\\
& \thirdinequal 2 \eta(t) \sum\nolimits_{m=1}^{M} \sum\nolimits_{i=2}^{\tau} \frac{B_m}{B} \mathbb{E} \left[ F_m (\boldsymbol{\theta}^*) - F_m(\boldsymbol{\theta}_m^i (t)) - \frac{\mu}{2} \left\| \boldsymbol{\theta}_m^i (t) - {\boldsymbol{\theta}}^* \right\|_2^2 \right]\nonumber\\
& = 2 \eta(t) \sum\nolimits_{m=1}^{M} \sum\nolimits_{i=2}^{\tau} \frac{B_m}{B} \mathbb{E} \left[ F_m (\boldsymbol{\theta}^*) - F_m^* + F_m^* - F_m(\boldsymbol{\theta}_m^i (t)) - \frac{\mu}{2} \left\| \boldsymbol{\theta}_m^i (t) - {\boldsymbol{\theta}}^* \right\|_2^2 \right]\nonumber\\
& = 2 \eta(t) (\tau - 1) \Gamma +  2 \eta (t) \sum\nolimits_{m=1}^{M} \sum\nolimits_{i=2}^{\tau} \frac{B_m}{B} \left( F_m^* - \mathbb{E} \left[ F_m({\boldsymbol{\theta}}_m^i (t)) \right] \right) \nonumber\\
& \quad - \mu \eta(t) \sum\nolimits_{m=1}^{M} \sum\nolimits_{i=2}^{\tau} \frac{B_m}{B} \mathbb{E} \left[ \left\| \boldsymbol{\theta}_m^i (t) - {\boldsymbol{\theta}}^* \right\|_2^2 \right],  
\end{align}
where (a) follows since $\mathbb{E}_{\xi} \left[ \nabla F_m \left( \boldsymbol{\theta} (t), \xi_m^i (t) \right) \right] = \nabla F_m \left( \boldsymbol{\theta} (t)  \right)$, $\forall i, m, t$, and (b) follows due to the fact that $F_m$ is $\mu$-strongly convex. We have
\begin{align}\label{EQ_LemmaTermE_app_2}
& - \left\| \boldsymbol{\theta}_m^i (t) - {\boldsymbol{\theta}}^* \right\|_2^2 = - \left\| \boldsymbol{\theta}_m^i (t) - {\boldsymbol{\theta}}_m^1 (t) \right\|_2^2 - \left\| {\boldsymbol{\theta}}_m^1 (t) - {\boldsymbol{\theta}}^* \right\|_2^2 - 2 \langle \boldsymbol{\theta}_m^i (t) - {\boldsymbol{\theta}}_m^1 (t) , {\boldsymbol{\theta}}_m^1 (t) - {\boldsymbol{\theta}}^* \rangle \nonumber\\
& \quad \firstinequal - \left\| \boldsymbol{\theta}_m^i (t) - {\boldsymbol{\theta}}_m^1 (t) \right\|_2^2 - \left\| {\boldsymbol{\theta}}_m^1 (t) - {\boldsymbol{\theta}}^* \right\|_2^2 + \frac{1}{\eta(t)} \left\| \boldsymbol{\theta}_m^i (t) - {\boldsymbol{\theta}}_m^1 (t) \right\|_2^2 + \eta(t) \left\| {\boldsymbol{\theta}}_m^1 (t) - {\boldsymbol{\theta}}^* \right\|_2^2 \nonumber\\
&\quad  = - (1 - \eta(t)) \left\| {\boldsymbol{\theta}}_m^1 (t) - {\boldsymbol{\theta}}^* \right\|_2^2 + \Big(\frac{1}{\eta(t)} - 1\Big) \left\| \boldsymbol{\theta}_m^i (t) - {\boldsymbol{\theta}}_m^1 (t) \right\|_2^2, \quad i \in [\tau], m \in [M],
\end{align}
where (a) follows from Cauchy-Schwarz inequality. For $\eta(t) \le 1$, we have
\begin{align}\label{Theta_m_1_theta_star}
& - (1 - \eta(t)) \mathbb{E} \left[ \left\| {\boldsymbol{\theta}}_m^1 (t) - {\boldsymbol{\theta}}^* \right\|_2^2 \right] = - (1 - \eta(t)) \mathbb{E} \left[ \left\| {\boldsymbol{\theta}} (t) + \widetilde{\boldsymbol{z}}_m^{\rm{dl}} (t) - {\boldsymbol{\theta}}^* \right\|_2^2 \right] \nonumber\\
& \quad \; =  - (1 - \eta(t)) \Big( \mathbb{E} \left[ \left\| {\boldsymbol{\theta}} (t) - {\boldsymbol{\theta}}^* \right\|_2^2 \right] + \mathbb{E} \left[ \left\| \widetilde{\boldsymbol{z}}_m^{\rm{dl}} (t) \right\|_2^2 \right] + \mathbb{E} \left[ 2 \langle  \boldsymbol{\theta} (t) - {\boldsymbol{\theta}}^*, \widetilde{\boldsymbol{z}}_m^{\rm{dl}} (t) \rangle \right] \Big) \nonumber\\
& \quad \; \thirdequal  - (1 - \eta(t)) \left( \mathbb{E} \left[ \left\| {\boldsymbol{\theta}} (t) - {\boldsymbol{\theta}}^* \right\|_2^2 \right] + \mathbb{E} \left[ \left\| \widetilde{\boldsymbol{z}}_m^{\rm{dl}} (t) \right\|_2^2 \right] \right),
\end{align}
where (a) follows since $\mathbb{E} \left[ \widetilde{\boldsymbol{z}}_m^{\rm{dl}} (t) \right] = \boldsymbol{0}$, and the fact that ${\boldsymbol{\theta}} (t)$ is independent of $\widetilde{\boldsymbol{z}}_m^{\rm{dl}} (t)$, for $m \in [M]$. 
According to \eqref{EQ_LemmaTermE_app_2} and \eqref{Theta_m_1_theta_star}, it follows that, for $i \in [\tau]$, $m \in [M]$, 
\begin{align}\label{Theta_m_i_theta_star}
- \mathbb{E} \left[ \left\| \boldsymbol{\theta}_m^i (t) - {\boldsymbol{\theta}}^* \right\|_2^2 \right] \le & - (1 - \eta(t)) \mathbb{E} \left[ \left\| {\boldsymbol{\theta}} (t) - {\boldsymbol{\theta}}^* \right\|_2^2 \right] + \Big(\frac{1}{\eta(t)} - 1\Big) \mathbb{E} \left[ \left\| \boldsymbol{\theta}_m^i (t) - {\boldsymbol{\theta}}_m^1 (t) \right\|_2^2 \right] \nonumber\\
& - (1 - \eta(t)) \mathbb{E} \left[ \left\| \widetilde{\boldsymbol{z}}^{\rm{dl}}_m (t)  \right\|_2^2 \right].   
\end{align}
Substituting \eqref{Theta_m_i_theta_star} into \eqref{EQ_LemmaTermE_D_app_1} yields
\begin{align}\label{EQ_LemmaTermE_app_3}
& 2 \eta(t) \sum\nolimits_{m=1}^{M} \sum\nolimits_{i=2}^{\tau} \frac{B_m}{B} \mathbb{E} \left[ \langle \boldsymbol{\theta}^* - \boldsymbol{\theta}_m^i (t) , \nabla F_m \left( \boldsymbol{\theta}_m^i (t), \xi_m^i (t) \right) \rangle \right] \nonumber\\
& \; \le - \mu \eta(t) (1 - \eta(t)) (\tau - 1) \mathbb{E} \left[ \left\| {\boldsymbol{\theta}} (t) - {\boldsymbol{\theta}}^* \right\|_2^2 \right] + \mu (1 - \eta(t)) \eta^2(t) G^2 \frac{\tau (\tau-1)(2\tau-1)}{6} \nonumber\\
& \; \quad + 2 \eta(t) (\tau - 1) \Gamma - \mu \eta(t) \left( 1 - \eta (t) \right) (\tau - 1) \sum\nolimits_{m=1}^{M} \frac{B_m}{B} \mathbb{E} \left[ \left\| \widetilde{\boldsymbol{z}}^{\rm{dl}}_m (t)  \right\|_2^2 \right] \nonumber\\
& \; \quad + 2 \eta (t) \sum\nolimits_{m=1}^{M} \sum\nolimits_{i=2}^{\tau} \frac{B_m}{B} \left( F_m^* - \mathbb{E} \left[ F_m({\boldsymbol{\theta}}_m^i (t)) \right] \right), 
\end{align}
where we have used the inequality in \eqref{Theta_m_iTheta_m_1}.
Substituting \eqref{AppLemmaTemr_1_D_Eq_1_bound} and \eqref{EQ_LemmaTermE_app_3} into \eqref{EQ_LemmaTermE_app_1} yields
\begin{align}
& 2 \eta(t) \sum\nolimits_{m=1}^{M} \sum\nolimits_{i=2}^{\tau} \frac{B_m}{B} \mathbb{E} \left[ \langle \boldsymbol{\theta}^* - {\boldsymbol{\theta}} (t) , \nabla F_m \left( \boldsymbol{\theta}_m^i (t), \xi_m^i (t) \right) \rangle \right] \nonumber\\
& \le - \mu \eta(t) (1 - \eta(t)) (\tau - 1) \mathbb{E} \left[ \left\| {\boldsymbol{\theta}} (t) - {\boldsymbol{\theta}}^* \right\|_2^2 \right] + \left( 1+ \mu (1 - \eta(t)) \right) \eta^2(t) G^2 \frac{\tau (\tau-1)(2\tau-1)}{6} \nonumber\\
& \; \quad + 2 \eta(t) (\tau - 1) \Gamma  - \eta(t) (\tau - 1) \left( \mu - \eta(t) (\mu + 1) \right)  \sum\nolimits_{m=1}^{M} \frac{B_m}{B} \mathbb{E} \left[ \left\| \widetilde{\boldsymbol{z}}^{\rm{dl}}_m (t)  \right\|_2^2 \right] \nonumber\\
& \; \quad + \left( \eta^2 (t) +1 \right) \left( \tau - 1 \right) G^2 + 2 \eta (t) \sum\nolimits_{m=1}^{M} \sum\nolimits_{i=2}^{\tau} \frac{B_m}{B} \left( F_m^* - \mathbb{E} \left[ F_m({\boldsymbol{\theta}}_m^i (t)) \right] \right)\nonumber\\
& \firstinequal - \mu \eta(t) (1 - \eta(t)) (\tau - 1) \mathbb{E} \left[ \left\| {\boldsymbol{\theta}} (t) - {\boldsymbol{\theta}}^* \right\|_2^2 \right] + \left( 1+ \mu (1- \eta(t)) \right) \eta^2(t) G^2 \frac{\tau (\tau-1)(2\tau-1)}{6} \nonumber\\
& \; \quad + 2 \eta(t) (\tau - 1) \Gamma + \left( \eta^2 (t) +1 \right) \left( \tau - 1 \right) G^2+ 2 \eta (t) \sum\nolimits_{m=1}^{M} \sum\nolimits_{i=2}^{\tau} \frac{B_m}{B} \left( F_m^* - \mathbb{E} \left[ F_m({\boldsymbol{\theta}}_m^i (t)) \right] \right),
\end{align}
where (a) follows since $\eta(t) \le \frac{\mu}{\mu + 1}$. This completes the proof of Lemma \ref{LemmaTermE}.

\bibliographystyle{IEEEtran}
\bibliography{Report}

\end{document}